\documentclass[
twocolumn,
reprint,
superscriptaddress,
nofootinbib,
nobibnotes,
amsmath,
amssymb,
aps,
longbibliography,
prb,
showkeys,
]{revtex4-2}
 
\usepackage[latin1]{inputenc}
\usepackage{amsmath}
\usepackage{amsfonts}
\usepackage{hyperref}
\usepackage{amssymb}
\usepackage{graphicx}
\usepackage{hyperref}
\usepackage{colortbl}
\usepackage{subfigure}
\usepackage{svg}
\usepackage{amsthm}
\usepackage{thmtools}
\usepackage[normalem]{ulem}
\usepackage{dsfont}
\usepackage{circuitikz}
\usepackage{cases}

\usepackage[T1]{fontenc}
\usepackage{mathptmx}
\usepackage{bbold}
\usepackage{scalerel}
\usepackage{xcolor}

\usepackage{mathtools}

\declaretheorem[style=definition]{definition} %If you want your theorems to be counted per section instead of subsection, then just remove the sub from the numberwithin
\numberwithin{equation}{section}

\declaretheorem[style=plain]{lemma}

\declaretheorem[style=definition]{remark}

\newcommand{\hyref}[1]{\hyperref[#1]{\ref{#1}}}
\newcommand{\dd}{\mathrm{d}}

\newcommand{\orange}[1]

\renewcommand{\thesection}{\arabic{section}}
\renewcommand{\thesubsection}{.\arabic{subsection}}

\newcommand{\thenewsubsection}{\thesection.\arabic{subsection}}
\newcommand{\thenewsubsubsection}{\thesection\thesubsection.\arabic{subsubsection}}

\makeatletter
\def\@hangfrom@section#1#2#3{\@hangfrom{#1#2}#3}%\MakeTextUppercase{#3}}%
\def\@hangfroms@section#1#2{#1#2}%\MakeTextUppercase{#2}}%
\makeatother

\usepackage{titlesec}

\titleformat{\section}  % which section command to format
  {\fontsize{12}{14}\bfseries} % format for whole line
  {\thesection} % how to show number
  {1em} % space between number and text
  {\centering} % formatting for just the text
  [] % formatting for after the text

\titleformat{\subsection}  % which section command to format
  {\fontsize{10}{12}\bfseries} % format for whole line
  {\thenewsubsection} % how to show number
  {2em} % space between number and text
  {\centering} % formatting for just the text
  [] % formatting for after the text

\titleformat{\subsubsection}  % which section command to format
  {\fontsize{10}{12}\itshape} % format for whole line
  {\thenewsubsubsection} % how to show number
  {2em} % space between number and text
  {\centering} % formatting for just the text
  [] % formatting for after the text
\begin{document}

\title{A simple mathematical theory for Simple Volatile Memristors and their spiking circuits}

\author{T. M. Kamsma}
\affiliation{Mathematical Institute, Utrecht University, Budapestlaan 6, 3584 CD Utrecht, The Netherlands}
\affiliation{Institute for Theoretical Physics, Utrecht University,  Princetonplein 5, 3584 CC Utrecht, The Netherlands}
\author{R. van Roij}
\affiliation{Institute for Theoretical Physics, Utrecht University,  Princetonplein 5, 3584 CC Utrecht, The Netherlands}
\author{C. Spitoni}
\affiliation{Mathematical Institute, Utrecht University, Budapestlaan 6, 3584 CD Utrecht, The Netherlands}

\date{\today}
\begin{abstract}
In pursuit of neuromorphic (brain-inspired) devices, memristors (memory-resistors) have emerged as effective components for emulating neuronal circuitry. Here we formally define a class of Simple Volatile Memristors (SVMs) based on a simple conductance equation of motion from which we build a simple mathematical theory on the dynamics of isolated SVMs and SVM-based spiking circuits. Notably, SVMs include various fluidic iontronic devices that have recently garnered significant interest due to their unique quality of operating within the same medium as the brain. Specifically we show that symmetric SVMs produce non self-crossing current-voltage hysteresis loops, while asymmetric SVMs produce self-crossing loops. Additionally, we derive a general expression for the enclosed area in a loop, providing a relation between the voltage frequency and the SVM memory timescale. These general results are shown to materialise in physical finite-element calculations of microfluidic memristors. An SVM-based circuit has been proposed that exhibits all-or-none and tonic neuronal spiking. We generalise and analyse this spiking circuit, characterising it as a two-dimensional dynamical system. Moreover, we demonstrate that stochastic effects can induce novel neuronal firing modes absent in the deterministic case. Through our analysis, the circuit dynamics are well understood, while retaining its explicit link with the physically plausible underlying system.
\end{abstract}
\maketitle

\section{Introduction}

The rapid advancement and widespread deployment of computing devices, especially in the realm of Artificial Intelligence, have led to an exponential and unsustainable increase in energy consumption, posing a critical challenge for future advancements of computation \cite{jones2018information}. Neuromorphic computing, inspired by the human brain's remarkable capabilities and energy efficiency, is one of the promising pursuits to tackle this in contemporary computational research \cite{schuman2017survey,schuman2022opportunities,mehonic2022brain}. To this end, memristors (memory resistors) have emerged as promising candidates \cite{jeong2016memristors,zhu2020comprehensive} due to their analogous behaviour to biological synapses, the connections between neurons, and to neuronal ion channels that enable signal propagation within individual neurons \cite{sah2014brains}, also driving research in memristor based mathematical neuron models \cite{chua2013memristor,liu2024firing,li2023electrical,li2022application,li2023locally,xu2022modeling,hu2019dynamic,al2015memristors,li2024firing,bao2018symmetric}. Memristors are typically classified as non-volatile or volatile, where a non-volatile memristor will retain its altered conductance and a volatile memristor will dynamically relax back to its equilibrium state when external driving forces are removed \cite{caravelli2018memristors}. Among the diverse types of memristors, fluidic iontronic memristors have recently garnered significant interest \cite{yu2023bioinspired,khan2023advancement,hou2023learning,kim2023liquid,xie2022perspective}, operating in the same aqueous electrolyte environment as the brain and offering possibilities not only for multiple information carriers in parallel but also for chemical regulation and bio-integrability \cite{han2022iontronics}. This has led to proposals for artificial neuronal spiking circuits \cite{robin2021principles,kamsma2023iontronic,kamsma2024advanced}, advances in implementing chemical regulation \cite{Han2023IontronicDissolution,robin2023long,xiong2023neuromorphic,wang2024aqueous}, and demonstrations of neuromorphic \cite{kamsma2024brain} and logic \cite{emmerich2024nanofluidic} computations.

In this work, we introduce and analyse a general class of \textit{Simple Volatile Memristors} (SVMs). This class of SVMs not only includes various volatile fluidic memristors \cite{robin2023long,kamsma2023iontronic,kamsma2023unveiling,cervera2024modeling,kamsma2024brain}, but also memristors naturally appearing in plants \cite{markin2014analytical}. The work we present here is of a generic mathematical nature, however we also provide clear and explicit links to physical iontronic model systems by verifying our predictions with physical finite-element calculations of memristive microfluidic ion channels. These simulations solely incorporate the physical equations governing the device dynamics and include no explicit prior knowledge of the generic mathematical framework presented in this work. In our comparison to physical devices, our focus is specifically directed towards fluidic iontronic memristors as these devices are the topic of much recent research, as discussed above.
\begin{figure*}[ht]
		\centering
		\includegraphics[width=1\textwidth]{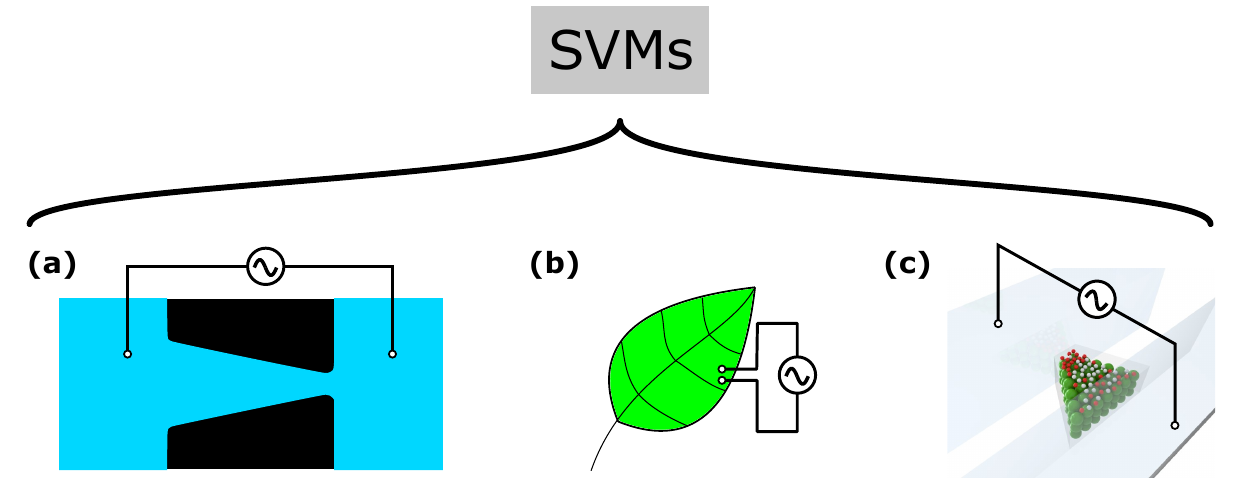}
		\caption{SVMs have been shown to include various distinct memristive devices. A few examples include \textbf{(a)} conical microchannels containing an aqueous electrolyte \cite{kamsma2023iontronic,kamsma2023unveiling,cervera2024modeling}, \textbf{(b)} naturally occurring memristors in plants \cite{markin2014analytical}, \textbf{(c)} fluidic channels containing a Nanochannel Network Membrane embedded within a colloidal structure \cite{choi2016high,kamsma2024brain}, and (sub)nanoscale planar slits (not depicted) \cite{robin2023long}.}
		\label{fig:IntroFig}
\end{figure*}

As various SVMs have garnered interest \cite{robin2023long,kamsma2023iontronic,kamsma2023unveiling,cervera2024modeling,kamsma2024brain,markin2014analytical}, we propose a simple unifying mathematical theory that consists of the simple conductance equation of motion for SVMs, from which we derive various dynamic features of the current-voltage hysteresis loops of individual SVMs and of SVM-based spiking circuits. We begin with formally defining the class of SVMs using the aforementioned conductance equation of motion in Sec.~\ref{sec:framework}, which is based on natural characteristics of generic volatile memristors. A standard feature of memristors is the pinched hysteresis loop that emerges in the current-voltage diagram when a periodic potential is applied \cite{chua2014if}. In Secs.~\ref{sec:typeIandII} and \ref{sec:area} we will derive various general properties of the hysteresis loops that emerge from SVMs. In Sec.~\ref{sec:typeIandII} we demonstrate a relation between the spatial symmetry of an SVM and the type of current-voltage hysteresis loops that emerge from it. Specifically, we establish that spatially symmetric SVMs consistently display hysteresis loops that are not self-crossing in the origin, in contrast to their asymmetric counterparts which produce self-crossing loops under reasonable assumptions. This prediction is shown to materialise in physical simulation results of symmetric and asymmetric memristive microchannels. In Sec.~\ref{sec:area}, we will derive a general expression for the area enclosed within the current-voltage hysteresis loop and show that the frequency of the periodic potential for which this area is maximal can be predicted when the steady-state conductance of an SVM is known. The enclosed area is (typically) maximal when $2\pi f\tau\approx1$, with $f$ the (dimensional) frequency of the applied (sinusoidal) voltage and $\tau$ the memory timescale of the device. For steady-state conductances that are linear functions of the voltage, the relation  $2\pi f\tau=1$ is exact. Since pinched hysteresis loops are a standard feature to investigate for memristors \cite{chua2014if}, our result provides a straightforward method to use experimentally observed hysteresis loops to determine the memory timescale $\tau$ of the device, a method already shown to be of value \cite{kamsma2024brain}. Conversely, if an estimate for $\tau$ is already known, then experiments can be sped up significantly by directly pinpointing the optimal frequency regime. 

Expanding our investigation to SVM based applications, we shift our attention in Sec.~\ref{sec:circuit} to a recently proposed fluidic iontronic circuit that exhibits characteristic features of neuronal communication in the form of all-or-none action potentials and spike trains \cite{kamsma2023iontronic}. We convert the originally four-dimensional system of equations to a two-dimensional dynamical system containing containing merely four parameters. This system is reminiscent of the FitzHugh-Nagumo \cite{fitzhugh1961impulses,nagumo1962active} and Morris-Lecar \cite{morris1981voltage} models, while being completely physically plausible due to its direct derivation from physical equations. Additionally, we will show that the inclusion of voltage noise, inherently present in any circuit \cite{sarpeshkar1993white}, enables new characteristic features of neuronal spiking that are not present in the deterministic case.

Given that various memristors were already found to be well described by the SVM class we define in Sec.~\ref{sec:framework} \cite{robin2023long,markin2014analytical,kamsma2023iontronic,cervera2024modeling,kamsma2023unveiling,kamsma2024brain}, it is reasonable to expect that new devices will be presented in the future that will also fall within our definition of SVMs. Therefore our work is also likely of relevance to future memristors and the spiking circuits they possible enable. Therefore, the results we derive here form a step towards a constructive frameworks for this class of memristors, which could be of particular value to the newly emerging field of fluidic iontronic neuromorphic devices.

\section{Simple Volatile Memristors}\label{sec:framework}
The dynamic conductance $g(t)$ of volatile memristors dynamically transitions to some steady-state conductance $h$ and reverts back to its equilibrium conductance if external driving forces are removed. In the case of a voltage-driven memristor, the steady-state conductance is determined by the voltage $V$ over the memristor, i.e.\ $h=h(V)$ \cite{chua2014if}. Since volatile memristors transition to their steady-state conductance for a given voltage, a natural assumption is that the time-derivative $\dot{g}$ of a volatile memristor equals some function $f$ of the difference $h(V(t))-g(t)$, i.e.\ 
\begin{align*}
    \dfrac{\dd g(t)}{\dd t}=f\big(h(V(t))-g(t)\big).
\end{align*}
For stability arguments, under the above assumption, it must hold that $f(0)=0$. However, system dynamics can be notoriously difficult physics problems, so the full form of $f$ is often unknown. Let us therefore expand $f$ to first order $f(x)=f(0)+f^{\prime}(0)x+\mathcal{O}(x^2)$, yielding 
\begin{align*}
   \dfrac{\dd g(t)}{\dd t}=f(0)+\frac{h(V(t))-g(t)}{\tau}+\mathcal{O}\big((h(V(t))-g(t))^2\big), 
\end{align*}
where $f^{\prime}(0)=\tau^{-1}$ and $\mathcal{O}\big((h(V(t))-g(t))^2\big)$ indicates higher order terms. The unknown factor $\tau$ must be some timescale on the basis of its dimensionality, which we can absorb in our time variable $t/\tau\to t$ to obtain a dimensionless time. With this we arrive at the simple expression $\dot{g}(t)=h(V(t))-g(t)$. Because this description entails volatile memristors with a simple decay towards some steady-state conductance $h(V(t))$, we name the class of memristor described by this method as Simple Volatile Memristors (SVMs). The full definition of an SVM is laid out below in Def.~\ref{def:SVMdef}.
\begin{definition}\label{def:SVMdef}
A \textit{Simple Volatile Memristor (SVM)} exhibits a dynamic conductance $g(t)$ that traces a steady-state conductance $h(V(t))$ such that the current through the SVM $I_{\mathrm{m}}(t)$ and the dynamic conductance upon applying a voltage $V(t)$ are described by
	\begin{numcases}{}
		I_{\mathrm{m}}(t)=g(t)V(t)\label{eq:Im},\\
		\dfrac{\mathrm{d}g(t)}{\mathrm{d}t}=h(V(t))-g(t)\label{eq:dgdt},
	\end{numcases}
with $h(x)$ an analytical function which, without loss of generality,  we consider to be of the form
\begin{align}
    h(x)=\left(1+\sum_{i=1}^{\infty}\alpha_{2i-1}x^{2i-1}+\sum_{j=1}^{\infty}\beta_{2j}x^{2j}\right),\label{eq:h}
\end{align}
with $x\in\mathds{R}$, and $\alpha_{i},\beta_{j}\in\mathds{R}$ for all $i,j\in\mathbb{N}$.
\end{definition}
\noindent It is straightforward to check that Def.~\ref{def:SVMdef} forms a subclass of the general definition of voltage-driven memristors \cite{chua2014if}.

The formulation of Def.~\ref{def:SVMdef} has been successfully applied to quantitatively describe various fluidic iontronic memristors \cite{robin2023long,kamsma2023iontronic,kamsma2023unveiling,cervera2024modeling,kamsma2024brain}. Def.~\ref{def:SVMdef} has also been applied in terms of the dynamic resistance, rather than the conductance, in the modelling of memristors that naturally appear in plants \cite{markin2014analytical}. Thus far most physical devices are described through their conductance \cite{robin2023long,kamsma2023iontronic,kamsma2023unveiling,cervera2024modeling,kamsma2024brain}, but Def.~\ref{def:SVMdef} and the consequent results we describe are straightforward to expand to also formally include dynamic resistances. Therefore, depending on the form of $h(V(t))$, SVMs include a variety of memristors. It should be noted that ordinarily the dynamic variable of a memristor is some internal state variable which then in turn determines the conductance, while Def.~\ref{def:SVMdef} treats the conductance itself directly as the dynamic variable. Within our SVM framework, it is in fact also possible that there is an internal physical state parameter which is equal to the conductance up to a multiplicative factor \cite{kamsma2023iontronic,kamsma2024brain}, meaning that Def.~\ref{def:SVMdef} still applies while the actual dynamical variable is some internal (physical) state of the memristor.

Via Eq.~(\ref{eq:h}) we can make a link to the spatial symmetry of the SVM. A symmetric SVM must exhibit the same equation of state $h(V)$ for a potential of $V$ and $-V$ (i.e.\ $h(V)=h(-V)$), since changing the sign is equivalent to switching the two terminals of the SVM. Since the sign of the potential can have no impact, we can conclude that all $\alpha_i=0$ for a symmetric SVM. This is how we will distinguish between symmetric and asymmetric SVMs as defined in Def.~\ref{def:symmetry}.
 \begin{definition}\label{def:symmetry}
     A \textit{symmetric SVM} is an SVM  as in Def.~\ref{def:SVMdef} such that $\alpha_{2i-1}=0$ as in Eq.~(\ref{eq:h}) $\forall i\in\mathds{N}$. An \textit{asymmetric SVM} is an SVM such that there exists an index $i\in\mathds{N}$ such that $\alpha_{2i-1}\neq0$ as in Eq.~(\ref{eq:h}). A \textit{simple asymmetric SVM} is an asymmetric SVM such that $h(x)\neq h(y)$ for all $x>0$ and $y<0$.
 \end{definition}
 \noindent We note that asymmetric physical SVMs presented so far are often also simple asymmetric SVMs \cite{markin2014analytical,kamsma2023iontronic,kamsma2023unveiling,kamsma2024brain}.

In experiments, a periodic sweeping potential $V(t)$ is typically imposed which creates a pinched hysteresis loop in the current-voltage ($I-V$) plane, the hallmark of a memristor \cite{chua2014if}. In Def.~\ref{def:voltage} we define the class of potentials $V(t)$ that includes virtually all signalling waveforms used to create hysteresis loops.
 \begin{definition}\label{def:voltage}
     A \textit{sweeping potential} $V:\mathds{R}\to\mathds{R}$ is a $T$-periodic smooth bounded function with bounded derivatives, with $V(t)=0$ exactly twice per period at times $t_1$ and $t_2=t_1+T/2$ and which is differentiable in $t_1$ and $t_2$. An \textit{antisymmetric sweeping potential} additionally satisfies that $V(t)=-V(t+T/2)$ for all $t$.
 \end{definition}
Although square and sawtooth waveforms are not differentiable in both $t_1$ and $t_2$, they can be approximated arbitrarily closely by smooth harmonics. Therefore Def.~\ref{def:voltage} effectively includes most typical signalling waveforms such as sine, triangle, square, and sawtooth waves.

The dynamics conductance and current in Def.~\ref{def:SVMdef} can be written in an integral form, where we can fix the integration constant using the periodicity of a sweeping potential. Eq.~(\ref{eq:dgdt}) has the general solution
\begin{align}\label{eq:intsol}
	g(t)=e^{-t}\int_{0}^{t}h(V(s))e^{s}\mathrm{d}s+e^{-t}g(0),
\end{align}
where $g(0)$ is an integration constant. The periodicity of a sweeping potential $V(t)$ results (after all transients have decayed) in a periodic conductance, i.e.\ $g(t)=g(t+T)$ \cite{yang2019periodic}. By using that $g(0)=g(T)$ we calculate
\begin{gather*}
	g(0)=e^{-T}\left[\int_{0}^{T}h(V(s))e^{s}\mathrm{d}s+g(0)\right],\\
	g(0)=\frac{\int_{0}^{T}h(V(s))e^{s}\mathrm{d}s}{e^{T}-1}.
\end{gather*}
Now that we have expressions for $g(t)$ and $g(0)$ (in the case of a periodic $V(t)$), we can calculate the current through the memristor $I_{\mathrm{m}}(t)=g(t)V(t)$ with the following expression
\begin{align}\label{eq:Imeq}
	I_{\mathrm{m}}(t)=V(t)e^{-t}\left[\int_{0}^{t}h(V(s))e^{s}\mathrm{d}s+g(0)\right].
\end{align}
Once we know the function $h$, i.e.\ the steady-state conductance, we can evaluate Eqs.~(\ref{eq:intsol}) and (\ref{eq:Imeq}) to calculate the dynamic conductance $g(t)$ and current $I_{\mathrm{m}}(t)$ for a $V(t)$ of choice.

\subsection{Connection to a physical SVM}
To make the connection to a physical SVM example explicit, we link our general model to (conical) microfluidic ion channels, iontronic model systems that have been extensively studied experimentally \cite{cheng2007rectified, siwy2006ion, bush2020chemical, jubin2018dramatic, siwy2002rectification,fulinski2005transport,siwy2005asymmetric} as well as numerically \cite{duleba2022effect, lan2016voltage, vlassiouk2008nanofluidic, liu2012surface, kubeil2011role}, with several analytical descriptions for the static properties of the channel \cite{boon2021nonlinear, dal2019confinement, poggioli2019beyond,uematsu2022analytic}. Specifically, we consider axisymmetric tapered channels of length $L$ with a charged surface connecting two electrolyte reservoirs. Recently the analytical understanding of such conical channels was extended to its conductance dynamics \cite{kamsma2023iontronic}, for which the framework laid out in Def.~\ref{def:SVMdef} was used.

Ref.~\cite{kamsma2023iontronic} focuses on a physical device and hence the used quantities are in physical and not dimensionless units. Here we delineate how to transform the dimensionless mathematical results into the relevant physical dimensions corresponding to a conical channel, making the practical applicability and interpretation of the results presented in this work explicit. The dynamical conductance of a single conical channel was found to be well described by $\tau\dot{g}_{\mathrm{s}}(t)=g_0h_{\infty}(V_{\mathrm{s}}(t))-g_{\mathrm{s}}(t)$, describing a relaxation to the steady-state conductance $g_0h_{\infty}(V_{\mathrm{s}}(t))$ depending on the voltage $V_{\mathrm{s}}(t)$ over the channel over a timescale $\tau$, with $g_0$ being the steady-state Ohmic conductance of the channel when no voltage is applied. We see that it is straightforward to write this equation in dimensionless form in agreement with Def.~\ref{def:SVMdef} by converting to dimensionless time $t/\tau$, dimensionless conductance $g_{\mathrm{s}}/g_0$, and dimensionless voltage $V/V_{\mathrm{r}}$, with $V_{\mathrm{r}}=1$ V a reference voltage indicating a typical voltage scale that emerges from the circuit presented in Ref.~\cite{kamsma2023iontronic}.

\section{The (a)symmetry of SVMs}\label{sec:typeIandII}
Memristors are characterized by an emerging hysteresis loop in the current-voltage ($I-V$) plane when a periodic potential is applied that is pinched in the origin \cite{chua2014if}. At this pinch, the loop can be self-crossing or not. A memristor that produces a self-crossing loop is named a type I memristor and a type II memristor produces a loop which does not cross itself, but rather has opposite tangent trajectories at the pinch \cite{pershin2011memory}. Here we show that whether an SVM is of type I or II depends on the symmetry of the device. With the help of the tangent vector $\vec{\gamma}(t)$ of the hysteresis loop in the $I-V$ plane we can investigate these classifications. This tangent vector is given by
\begin{align*}
	\vec{\gamma}(t)=\begin{pmatrix}
		\dfrac{\mathrm{d}I_{\mathrm{m}}(t)}{\mathrm{d}t}\vspace{0.2 cm}\\
		\dfrac{\mathrm{d}V(t)}{\mathrm{d}t}
	\end{pmatrix}.
\end{align*}
Let us first give some more rigorous definitions of type I and type II memristors below in Def.~\ref{def:typeDef}.
\begin{definition}\label{def:typeDef}
	Let $V(t)$ be a sweeping potential as in Def.~\ref{def:voltage}.
	\begin{itemize}
		\item An SVM is of \textit{type I} if and only if $\forall\lambda\in\mathbb{R}\backslash \{0\}$ it holds that $\vec{\gamma}(t_1)\neq\lambda\vec{\gamma}(t_2)$.
		\item An SVM is of \textit{type II} if and only if $\exists\lambda\in\mathbb{R}\backslash \{0\}$ such that $\vec{\gamma}(t_1)=\lambda\vec{\gamma}(t_2)$.
	\end{itemize}
\end{definition}
Note that the class of periodic potentials $V(t)$ as defined in Def.~\ref{def:typeDef} includes any harmonic waveform that has precisely one distinct positive area and precisely one distinct negative area per period. 

\subsection{Relation between type of SVMs and their symmetry}\label{sec:timeorientation}
\begin{figure*}[ht]
		\centering
		\includegraphics[width=1\textwidth]{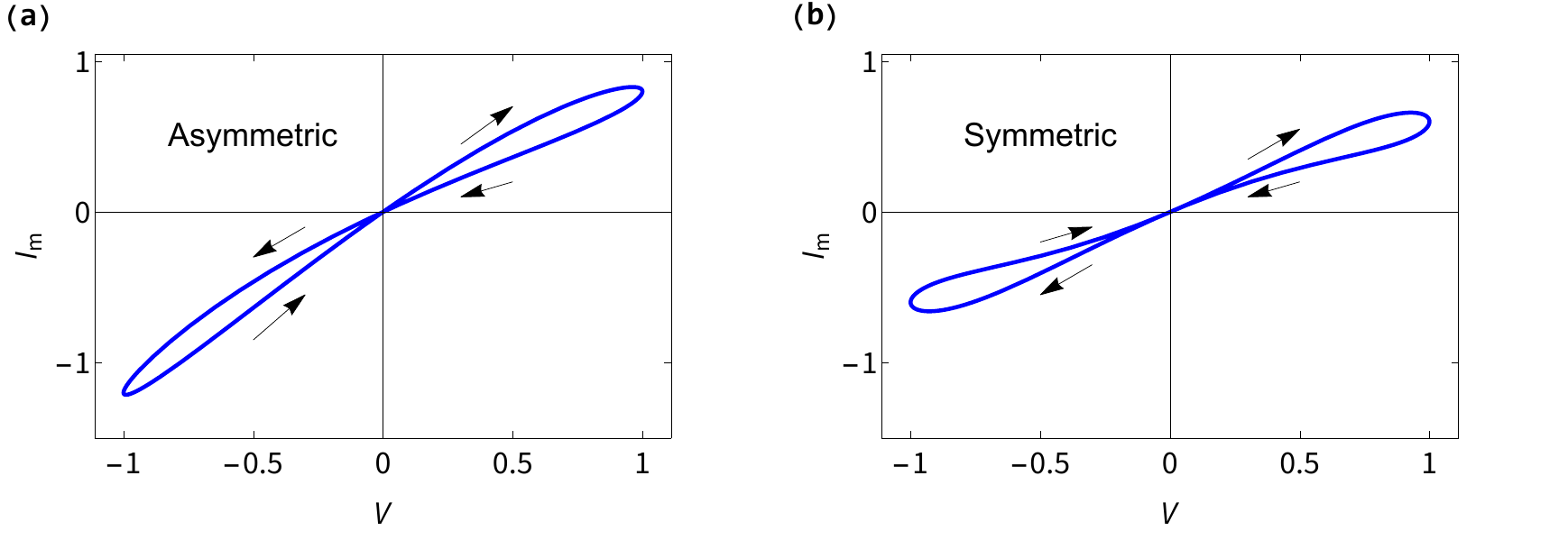}
		\caption{$I-V$ diagrams of SVMs with $V(t)=\sin(\omega t)$ with $\omega=2\pi/T$, $T=2\pi$, and $h(V(t))=1+\alpha V(t)+\beta V(t)^2$ for \textbf{(a)} $\alpha =-2/3$, $\beta =0$ and \textbf{(b)} $\alpha =0$, $\beta =-2/3$. This corresponds to \textbf{(a)} a simple asymmetric SVM and \textbf{(b)} a symmetric SVM, where we see the asymmetric SVM is of type I in agreement with Lemma \ref{lem:asym} and that the symmetric SVM is of type II in agreement with Lemma \ref{lem:sym}.}
		\label{fig:symasymfig}
\end{figure*}
In the next two lemmas we prove that Symmetric SVMs are of type II and that simple asymmetric SVMs are of type I.

\begin{lemma}\label{lem:sym}
	Let $V(t)$ be an antisymmetric sweeping potential as in Def.~\ref{def:voltage}, then a symmetric SVM as in Def.~\ref{def:symmetry} is of type II.
\end{lemma}
\begin{proof}
	Let us start by calculating $\vec{\gamma}(t)$ in $t_1$ and $t_2$, which yields
	\begin{align*}
	\vec{\gamma}(t_1)=\begin{pmatrix}
		g(t_1)\dot{V}(t_1)\vspace{0.2 cm}\\
		\dot{V}(t_1)
	\end{pmatrix},\quad\vec{\gamma}(t_2)=\begin{pmatrix}
		g(t_2)\dot{V}(t_2)\vspace{0.2 cm}\\
		\dot{V}(t_2)
	\end{pmatrix},
\end{align*}
where we used that $V(t_1)=V(t_2)=0$. Since $\alpha_i=0$ for all $i$ in Eq.~(\ref{eq:h}), $h$ is $T/2$-periodic as $h$ solely contains even powers of $V(t)$ and $V(t)=-V(t+T/2)$ for all $t$ for antisymmetric sweeping potentials. So Eq.~(\ref{eq:dgdt}) is a nonautonomous ordinary differential equation with $T/2$-periodic parameters, therefore there exists a unique, $T/2$-periodic solution $g(t)$ \cite{yang2019periodic}. Since $g(t)$ is $T/2$-periodic and $t_1=t_2+T/2$, it follows directly that $g(t_1)=g(t_2)$. Now choose $\lambda=\dot{V}(t_1)/\dot{V}(t_2)$, from which it follows that $\vec{\gamma}(t_1)=\lambda\vec{\gamma}(t_2)$ and we conclude that the SVM is of type II.
\end{proof}

\begin{lemma}\label{lem:asym}
	Let $V(t)$ be a sweeping potential as in Def.~\ref{def:voltage}, then a simple asymmetric SVM as in Def.~\ref{def:symmetry} is of type I.
\end{lemma}
\begin{proof}
For all $t_{-}\in(t_1-T/2,t_1)$ and all $t_{+}\in(t_1,t_1+T/2)$ we know that $\text{sgn}(V(t_{+}))=-\text{sgn}(V(t_{-}))$. Therefore either $h(V(t_{+}))<h(V(t_{-}))$, or $h(V(t_{+}))>h(V(t_{-}))$, for all $t_{-}\in(t_1-T/2,t_1)$ and all $t_{+}\in(t_1,t_1+T/2)$. With this in mind we can define the following integrals $G_{+}$ and $G_{-}$ and conclude that they must be unequal
\begin{align*}
    G_{-}=&\int_{t_1-T/2}^{t_1}h(V(s))e^{s}\mathrm{d}s,\\
    G_{+}=&\int_{t_1-T/2}^{t_1}h(V(s+T/2))e^{s}\mathrm{d}s\neq G_{-}.
\end{align*}
We then calculate $g(t_2-T)$ and $g(t_2)$
    \begin{align*}
        g(t_2-T)=&e^{-t_1+T/2}\int_{0}^{t_1-T/2}h(V(s))e^{s}\mathrm{d}s+e^{-t_1+T/2}g(0)\\
        =&e^{-t_1+T/2}\left[\int_{0}^{t_1}h(V(s))e^{s}\mathrm{d}s-\int_{t_1-T/2}^{t_1}h(V(s))e^{s}\mathrm{d}s\right]\\
        &+e^{-t_1+T/2}g(0)\\
        =&e^{T/2}g(t_1)-e^{-t_1+T/2}g(0)\\
        &-e^{-t_1+T/2}\int_{t_1-T/2}^{t_1}h(V(s))e^{s}\mathrm{d}s+e^{-t_1+T/2}g(0)\\
        =&e^{T/2}g(t_1)-e^{-t_1+T/2}G_{-}\\\\
        \end{align*}
A similar calculation for $g(t_2)$ yields
        \begin{align*}
        g(t_2)=&e^{-t_1-T/2}\int_{0}^{t_1+T/2}h(V(s))e^{s}\mathrm{d}s+e^{-t_1-T/2}g(0)\\
        =&e^{-t_1-T/2}\left[\int_{0}^{t_1}h(V(s))e^{s}\mathrm{d}s+\int_{t_1}^{t_1+T/2}h(V(s))e^{s}\mathrm{d}s\right]\\
        &+e^{-t_1-T/2}g(0)\\
        =&e^{-T/2}g(t_1)+e^{-t_1-T/2}\int_{t_1-T/2}^{t_1}h(V(s+T/2))e^{s+T/2}\mathrm{d}s\\
        =&e^{-T/2}g(t_1)+e^{-t_1}G_{+}\\
    \end{align*}
By then adding $g(t_2-T)e^{-T/2}$ and $g(t_2)$ we obtain
    \begin{align*}
        g(t_2)+g(t_2-T)e^{-T/2}=&g(t_1)-e^{-t_1}G_{-}+e^{-T/2}g(t_1)+e^{-t_1}G_{+}.
    \end{align*}
Since $h(V(t))$ is $T$-periodic, so is $g(t)$, thus $g(t_2-T)=g(t_2)$
    \begin{align*}
        g(t_2)(1+e^{-T/2})=&g(t_1)(1+e^{-T/2})+e^{-t_1}\left(G_{+}-G_{-}\right)
    \end{align*}
As we saw in the beginning of this proof $G_{+}\neq G_{-}$ so
\begin{align*}
    e^{-t_1}\left(G_{+}-G_{-}\right)\neq 0 \implies g(t_1)\neq g(t_2).
\end{align*}
We already calculated $\vec{\gamma}(t)$ in $t_1$ and $t_2$ in the proof of Lemma \ref{lem:sym}. It is clear that $\forall\lambda\in\mathbb{R}\backslash\{0\}$ it holds that $\vec{\gamma}(t_1)\neq\lambda\vec{\gamma}(t_2)$ if $g(t_1)\neq g(t_2)$. We conclude that a simple asymmetric SVM as in Def.~\ref{def:symmetry} is of type I.
\end{proof}
Lemmas \ref{lem:sym} and \ref{lem:asym} nicely explain the observation on how different iontronic memristors produce different $I-V$ hysteresis loops remarked in Ref.~\cite{noy2023fluid}, where we now identify the device (a)symmetry as the relevant discerning factor (provided the iontronic device is an SVM). In the case of Lemma \ref{lem:sym} the result is not surprising, as it is intuitive that a symmetric device will yield point symmetric (type II) $I-V$ diagrams for antisymmetric sweeping potentials, whereas Lemma \ref{lem:asym} provides some less trivial criteria for the case of when an SVM is of type I (see, e.g.\ Remarks \ref{rk:h} and \ref{rk:asym}). In Sec.~\ref{sec:comparisonIV} we provide a physical example of Lemmas \ref{lem:sym} and \ref{lem:asym} by evaluating numerical simulations of (a)symmetric fluidic channels, showing that they indeed are of type (I) II.

\begin{remark}\label{rk:h}
     In the case of antisymmetric sweeping potentials, to prove Lemma \ref{lem:asym} the constraint for simple asymmetric SVMs on $h$ that $h(x)\neq h(y)$ for $x>0$ and $y<0$  can be relaxed to either of the two alternative constraints:
     \begin{enumerate}
         \item $h(x)\neq h(-x)$, for all $x>0$.
         \item There is only one $i\in \mathbb{N}$ such that  $\alpha_i\neq0$.
     \end{enumerate}
      These alternative constraints would also ensure that $G_{+}\neq G_{-}$, the rest of the proof goes the same. This is relevant since various typical signalling waveforms, such as sine, triangle, and square waves are also antisymmetric sweeping potentials.
\end{remark}
\begin{remark}\label{rk:asym}
    Not all asymmetric SVMs will be self-crossing in the origin. For example, we can pick specific $V(t)$, $\alpha_1$, and $\alpha_3$ such that the corresponding terms precisely cancel out in the integrals $G_{+}$ and $G_{-}$, which implies that $g(t_1)=g(t_2)$. However this would require rather convoluted steady-state conductances $h(V)$, not often reported in the physical SVMs presented thus far. 
\end{remark}
\begin{figure*}[ht]
		\centering
		\includegraphics[width=1\textwidth]{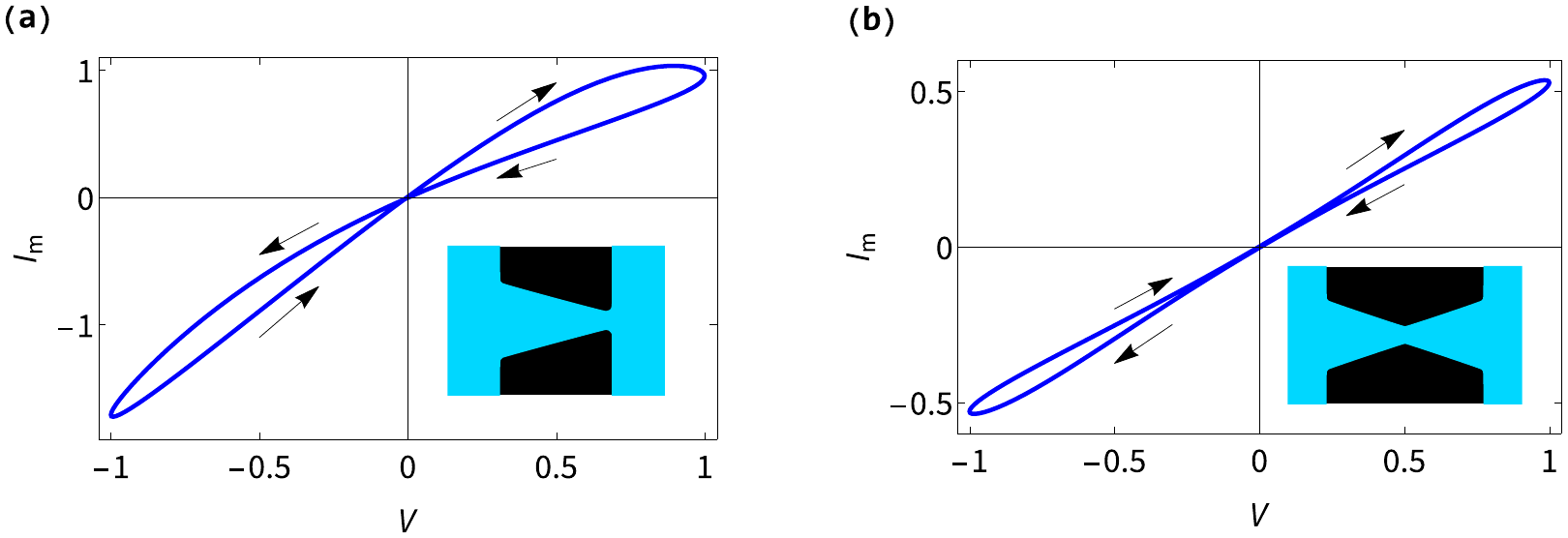}
		\caption{$I-V$ diagrams from physical continuum transport FE calculations of \textbf{(a)} a conical channel and \textbf{(b)} an hourglass channel when a periodic potential $V(t)=\sin(\omega t)$ is applied, with $\omega=2\pi f$ and $f=25$ Hz. These types of microfluidic channels were shown to be SVMs \cite{kamsma2023iontronic,kamsma2023unveiling}. The set of physical transport equations used for the results are described in the Appendix. Due to their respective geometries these are \textbf{(a)} an asymmetric SVM and \textbf{(b)} a symmetric SVM. The asymmetric conical channel is of type I in agreement with Lemma \ref{lem:asym} and the symmetric hourglass channel is of type II in agreement with Lemma \ref{lem:sym}.}
		\label{fig:conehourglass}
\end{figure*}

\subsection{Solving an SVM example}\label{sec:SVMexamp}
Let us calculate the memristor current $I_{\mathrm{m}}(t)$ for example SVMs of the form $h(x)=1+\alpha x+\beta x^2$ and show that they are indeed of type I if $\alpha\neq 0$ and of type II if $\alpha=0$. Consider the sweeping potential $V(t)=\sin(\omega t)$ with $t_1=nT$ and $t_2=(n+1/2)T$, for some $n\in\mathds{N}$. In this case $\vec{\gamma}(t_1)=\lambda\vec{\gamma}(t_2)$ if and only if $\lambda=-1$. So all that is left now to see if $\vec{\gamma}(t_1)\parallel \vec{\gamma}(t_2)$ is check whether $\dot{I}_{\mathrm{m}}(t_1)=-\dot{I}_{\mathrm{m}}(t_2)$. If this is the case then our example SVM is of type II, otherwise it is of type I. Using Eq.~(\ref{eq:intsol}) we find the following explicit expression for $I_{\mathrm{m}}(t)$
\begin{align*}
	I_{\mathrm{m}}(t)=&\sin(\omega t)\left[1+\alpha\frac{\left(\sin(\omega t)-\omega\cos(\omega t)\right)}{1+\omega^2}\right.\\
 &\left.+\frac{1}{2}\beta -\frac{1}{2}\beta \frac{\cos(2\omega t)+2\omega\sin(2\omega t)}{1+(2\omega)^2}\right].
\end{align*}
A straightforward calculation of $\dot{I}_{\mathrm{m}}$ evaluated at $t_1=nT$ and $t_2=(n+1/2)T$ then yields
\begin{align*}
	\left.\dfrac{\mathrm{d}I_{\mathrm{m}}(t)}{\mathrm{d}t}\right|_{t=t_1}=&\omega\left[1-\alpha\frac{\omega}{1+\omega^2}+\frac{1}{2}\beta -\frac{1}{2}\beta \frac{1}{1+(2\omega)^2}\right],\\
	\left.\dfrac{\mathrm{d}I_{\mathrm{m}}(t)}{\mathrm{d}t}\right|_{t=t_2}=&-\omega\left[1+\alpha\frac{\omega}{1+\omega^2}+\frac{1}{2}\beta -\frac{1}{2}\beta \frac{1}{1+(2\omega)^2}\right].
\end{align*}
As detailed, we know that the loop is self crossing at the origin only if $\dot{I}_{\mathrm{m}}(t_1)+\dot{I}_{\mathrm{m}}(t_2)=0$. We find that
\begin{align*}
	\left.\dfrac{\mathrm{d}I_{\mathrm{m}}(t)}{\mathrm{d}t}\right|_{t=t_1}+\left.\dfrac{\mathrm{d}I_{\mathrm{m}}(t)}{\mathrm{d}t}\right|_{t=t_2}=-\alpha\frac{2\omega^2}{1+\omega^2}.
\end{align*}
Since $\omega>0$, we see that $\vec{\gamma}(t_1)+\vec{\gamma}(t_2)=0$ if and only if $\alpha = 0$. Thus we conclude that indeed for any $\alpha\neq0$, the corresponding SVM is of type I, while it is of type II if $\alpha=0$. In Figs.~\ref{fig:symasymfig}(a) and \ref{fig:symasymfig}(b) we see $I-V$ diagrams for $\beta=0$ and $\alpha=0$, respectively. Indeed we see a self-crossing hysteresis in Fig.~\ref{fig:symasymfig}(a) while a the hysteresis loop in Fig.~\ref{fig:symasymfig}(b) does not cross itself.

As a final notion we consider the high frequency limit, since any memristor should behave like an ohmic resistor when the periodic potential is of high frequency \cite{chua1976memristive} (compared to the typical system timescale). So the slope of the curve in the $I-V$ plane should always be equal for $t_1$ and $t_2$ if $\omega\to\infty$, even when $\alpha\neq 0$. The slope in the $I-V$ plane in the limit $\omega\to\infty$ at $t_1$ and $t_2$ for our example SVM is given by
\begin{align*}
	\lim\limits_{\omega\to\infty}\left.\left(\dfrac{\mathrm{d}I_{\mathrm{m}}(t)}{\mathrm{d}t}\right)/\left(\dfrac{\mathrm{d}V(t)}{\mathrm{d}t}\right)\right|_{t=t_1}=&1+\beta/2,\\
	\lim\limits_{\omega\to\infty}\left.\left(\dfrac{\mathrm{d}I_{\mathrm{m}}(t)}{\mathrm{d}t}\right)/\left(\dfrac{\mathrm{d}V(t)}{\mathrm{d}t}\right)\right|_{t=t_2}=&1+\beta/2.
\end{align*}
So we indeed see that for $\omega\to\infty$, the slope for both $t_1$ and $t_2$ reduces to $1+\beta/2$, which is in agreement with the expectation of a generic memristor \cite{chua1976memristive}.

\subsection{Comparison to fluidic microchannel memristors}\label{sec:comparisonIV}
We claim that the SVMs we describe mathematically correspond to various physical devices \cite{robin2023long,markin2014analytical,kamsma2023iontronic,kamsma2023unveiling,kamsma2024brain}. Therefore, Lemmas \ref{lem:sym} and \ref{lem:asym} should materialise in such systems. To show this explicitly we will conduct physical continuum transport finite element (FE) calculations of (symmetric) hourglass and (asymmetric) conical microfluidic ion channels, using the FE analysis package COMSOL \cite{multiphysics1998introduction,pryor2009multiphysics}. The asymmetric conical channels are directly based on the channels discussed in Ref.~\cite{kamsma2023iontronic}, while the symmetric hourglass channels only differ in their geometry, but are otherwise subject to the same system parameters and physics. The FE calculations only incorporate the microscopic physical ion and fluid transport equations which contain no direct knowledge of any memristive properties, or channel conductance in general for that matter. Consequently, our FE calculations contain no prior knowledge whatsoever of the definitions and properties of SVMs shown in this work. In Fig.~\ref{fig:conehourglass}(a) we show the $I-V$ diagram resulting from a conical (and thus asymmetric) channel. A characteristic self-crossing hysteresis loop emerges, from which we conclude that the asymmetric channel indeed is of type I, as predicted by Lemma \ref{lem:asym}. Lemma \ref{lem:sym} then states that if the same system is converted to a symmetric channel, that the channel should form a type II memristor. Indeed, in Fig.~\ref{fig:conehourglass}(b) we see that the hysteresis loop that emerges from an hourglass (and thus symmetric) channel does not cross itself. The results shown in Fig.~\ref{fig:conehourglass} are direct physical manifestations of Lemmas \ref{lem:sym} and \ref{lem:asym}.

\subsection{Capacitor in parallel with a memristor}\label{sec:cap}
Def.~\ref{def:SVMdef} does not include any capacitance and thus we always have $I(t)=0$ if $V(t)=0$. In reality, physical systems can have an intrinsic capacitive element, which then materialises in a $I-V$ curve that does not pass through the origin \cite{sun2020non}. This is because a capacitive current is proportional to the time-derivative of the voltage over the capacitor, which is not necessarily zero when the voltage is zero. Capacitance has been reported for many different memristive devices \cite{markin2014analytical,abunahla2016resistive,zheng2018metal,kaur2018dopant,sarma2016observed,qingjiang2014memory,messerschmitt2015does,zhou2019evolution,xiao2019ultrathin,lee2016intrinsic}. Therefore, a more accurate circuit of a memristive device might be the circuit depicted in Fig.~\ref{fig:capcircuit}, where a memristor is connected in parallel with a series of a resistor with resistance $R$ and a capacitor with capacitance $C$. After transforming to dimensionless units, i.e.\ $V_{\mathrm{C}}/V_{\mathrm{r}}\to V_{\mathrm{C}}$, $I_{\mathrm{C}}R/V_{\mathrm{r}}\to I_{\mathrm{C}}$, $\tau_{\mathrm{RC}}=RC/\tau$, the dimensionless capacitive current is given by $I_{\mathrm{C}}=\tau_{\mathrm{RC}}\dot{V}_{\mathrm{C}}$, with $V_{\mathrm{C}}$ the voltage over the capacitor. The total current through the device is then $I_{\mathrm{m}}(t)+I_{\mathrm{C}}(t)$, meaning that the measured current will no longer be described by Def.~\ref{def:SVMdef}. However, most physical SVMs exhibit a vanishing current when $V(t)=0$, implying that their intrinsic capacitance is negligible \cite{robin2023long,kamsma2023iontronic,kamsma2023unveiling,kamsma2024brain}. Here we will show that in the limit of vanishing capacitance, i.e.\ $\tau_{\mathrm{RC}}\to0$, that $I_{\mathrm{C}}\to0$ and thus that Def.~\ref{def:SVMdef} still applies, despite the (negligible) capacitance.

\begin{figure}[t]
	\centering
	\begin{circuitikz} \draw
		(0,4)
		to[short,o-] (6,4) to[short,l=$I_{\mathrm{C}}$] (6,3 )to[capacitor,l=$C$] (6,2) to[resistor,l=$R$] (6,0)
		to[short,-o] (0,0)
		(3,4) to[short,l=$I_{\mathrm{m}}$] (3,3) to[memristor,l=$g$] (3,1) -- (3,0)
		;
	\end{circuitikz}
	\caption{Circuit including the capacitance of a memristive device \cite{markin2014analytical} by a conductive pathway with a capacitor with capacitance $C$ and a resistor with resistance $R$ connected in parallel to a pathway with a memristor with conductance $g$. Over the circuit a voltage $V$ can be applied, a current $I_{\mathrm{m}}(t)=g(t)V(t)$ flows through the pathway containing the memristor. A dimensionless capacitive current $I_{\mathrm{C}}(t)=\tau_{\mathrm{RC}}\dfrac{\mathrm{d}V_{\mathrm{C}}(t)}{\mathrm{d}t}$ flows through the capacitive pathway, where $V_{\mathrm{C}}(t)$ is the dimensionless potential over the capacitor and $\tau_{\mathrm{RC}}=RC/\tau$ the dimensionless RC time.}\label{fig:capcircuit}
\end{figure}
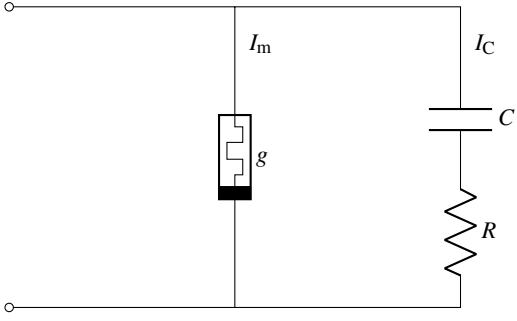

\begin{figure*}[ht]
		\centering
		\includegraphics[width=1\textwidth]{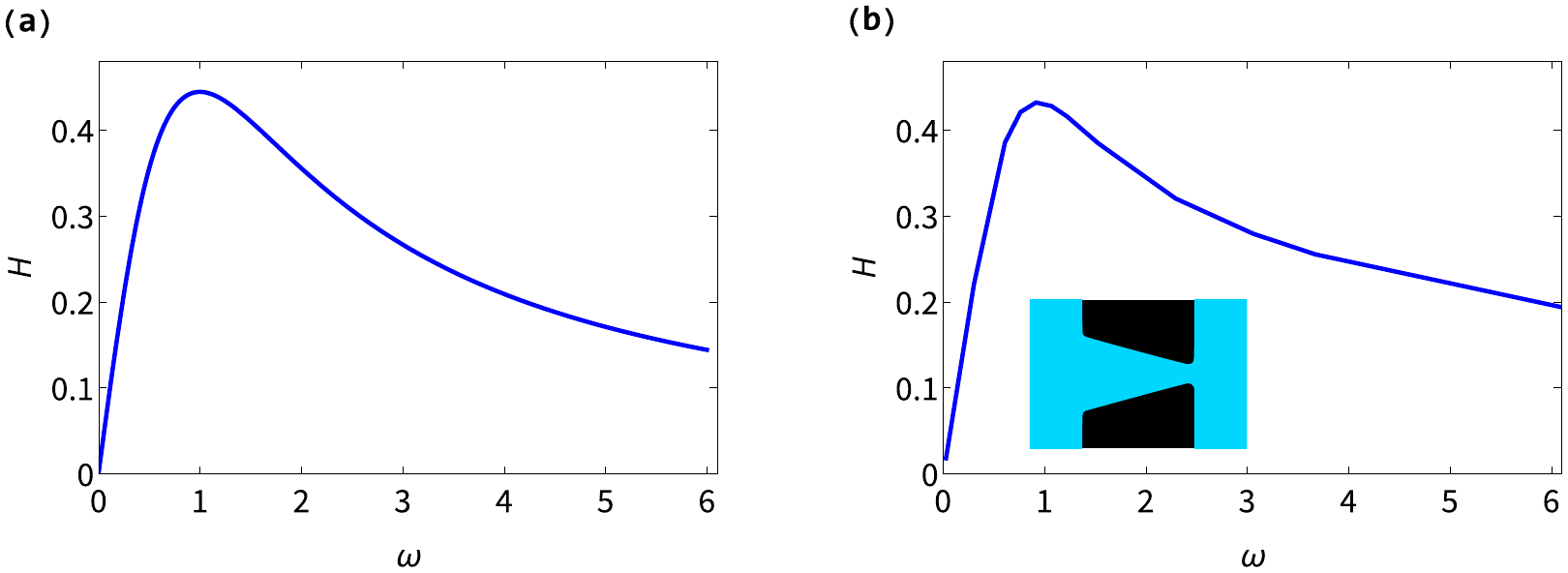}
		\caption{Enclosed surface area $H$ within an $I-V$ hysteresis loop of an SVM as a function of $\omega$ upon a sweeping potential $V(t)=\sin(\omega t)$ of \textbf{(a)} a general SVM with $h(x)=1+\alpha x$ and $\alpha =-2/3$, and \textbf{(b)} of a conical microfluidic ion channel determined by FE calculations of physical continuum transport equations.}
		\label{fig:Hfig}
\end{figure*}
In Ref.~\cite{markin2014analytical} it was shown that the capacitive current $I_{\mathrm{C}}(t)$ for the circuit in Fig.~\ref{fig:capcircuit} is given by
\begin{align}\label{eq:Iceq}
	I_{\mathrm{C}}(t)=&\dfrac{\mathrm{d}}{\mathrm{d}t}\left(e^{-t/\tau_{\mathrm{RC}}}\left[\int_{0}^{t}V(s)e^{s/\tau_{\mathrm{RC}}}\mathrm{d}s+V_{\mathrm{C}}(0)\right]\right),\\
 V_{\mathrm{C}}(0)=&\frac{\frac{1}{\tau_{\mathrm{RC}}}\int_{0}^{T}V(s)e^{s/\tau_{\mathrm{RC}}}\mathrm{d}s}{e^{T/\tau_{\mathrm{RC}}}-1},
\end{align}
where we again assume a sweeping potential $V(t)$ as per Def.~\ref{def:voltage}. We remark that if one does want to take capacitance into account, rather than ignoring it, that this is also possible by evaluating Eq.~(\ref{eq:Iceq}) once $C$ and $R$ are known.

Let us show that the capacitive current vanishes for $C\to0$, which is equivalent in our case to the limit $\tau_{\mathrm{RC}}\to0$. Firstly, note that $I_{\mathrm{m}}(t)$ is not dependent on $C$, so this quantity stays unchanged when we take the limit $\tau_{\mathrm{RC}}\to0$. The capacitive current $I_{\mathrm{C}}(t)$ is given by
\begin{align*}
	I_{\mathrm{C}}(t)=&-e^{-t/\tau_{\mathrm{RC}}}\left[\frac{1}{\tau_{\mathrm{RC}}}\int_{0}^{t}V(s)e^{s/\tau_{\mathrm{RC}}}\mathrm{d}s+V_{\mathrm{C}}(0)\right]+V(t).
\end{align*}
It is clear that the initial condition term $V_{\mathrm{C}}(0)$ will trivially vanish in the limit $\tau_{\mathrm{RC}}\to0$. Let us then focus on the first term
\begin{align*}
	-e^{-t/\tau_{\mathrm{RC}}}\frac{1}{\tau_{\mathrm{RC}}}\int_{0}^{t}V(s)e^{s/\tau_{\mathrm{RC}}}\mathrm{d}s=&-\int_{0}^{t}V(s)e^{-(t-s)/\tau_{\mathrm{RC}}}\frac{\mathrm{d}s}{\tau_{\mathrm{RC}}}
\end{align*}
We substitute $y=(t-s)/\tau_{\mathrm{RC}}$ and integrate by parts to obtain
\begin{align*}
    &-\int_{0}^{t/\tau_{\mathrm{RC}}}V(t-\tau_{\mathrm{RC}}y)e^{-y}\mathrm{d}y\\
    =&-V(t)+V(0)e^{-t/\tau_{\mathrm{RC}}}-\tau_{\mathrm{RC}}\int_{0}^{t/\tau_{\mathrm{RC}}}V^{\prime}(t-\tau_{\mathrm{RC}}y)e^{-y}\mathrm{d}y,
\end{align*}
where $V^{\prime}=\dd V/\dd y$. Now we need to show that the second and third term vanish in the limit $\tau_{\mathrm{RC}}\to 0$. The second term trivially vanishes, for the third term we find
\begin{align*}
    \tau_{\mathrm{RC}}\int_{0}^{t/\tau_{\mathrm{RC}}}V^{\prime}(t-\tau_{\mathrm{RC}}y)e^{-y}\mathrm{d}y\leq&||V^{\prime}||_{\infty} \int_{0}^{t/\tau_{\mathrm{RC}}}e^{-y}\mathrm{d}y\\
    =&||V^{\prime}||_{\infty}\tau_{\mathrm{RC}}\left[1-e^{-t/\tau_{\mathrm{RC}}}\right].
\end{align*}
This vanishes in the limit $\tau_{\mathrm{RC}}\to0$ since $V^{\prime}$ is bound and thus $||V^{\prime}||_{\infty}$ is finite. Thus we find that
\begin{align*}
    \lim\limits_{\tau_{\mathrm{RC}}\to0}-e^{-t/\tau_{\mathrm{RC}}}\frac{1}{\tau_{\mathrm{RC}}}\int_{0}^{t}V(s)e^{s/\tau_{\mathrm{RC}}}\mathrm{d}s=-V(t).
\end{align*}
Therefore we conclude that
\begin{align*}
     \lim\limits_{\tau_{\mathrm{RC}}\to0}\tau_{\mathrm{RC}}\dfrac{\mathrm{d}V_{\mathrm{C}}(t)}{\mathrm{d}t}=&-V(t)+V(t)=0.
\end{align*}
We see that indeed $\tau_{\mathrm{RC}}\dfrac{\mathrm{d}V_{\mathrm{C}}(t)}{\mathrm{d}t}\to 0$, thus the overall measured current is again given by $I_{\mathrm{m}}$ in agreement with Def.~\ref{def:SVMdef}.

\section{Quantitative measure of SVM \textit{I-V} hysteresis}\label{sec:area}
The area $H$ enclosed inside the hysteresis loop has been a property of interest \cite{georgiou2012quantitative,biolek2014interpreting,isah2022review}. Here we will present how the area can be used as a tool to estimate the typical timescale $\tau$ of an SVM, a method we already successfully applied in SVM experiments \cite{kamsma2024brain}. We present a general expression for $H$ for a type I SVM where the $I-V$ loop only intersects itself at the origin, for the case that a sweeping potential $V(t)=\sin(\omega t)$ is applied of period $T=2\pi/\omega$, where we recall that our dimensionless time $t$ is in units of the typical SVM memory retention time $\tau$. The enclosed surface area within a loop is given by
\begin{align*}
	H=&\left|\int_{nT}^{(n+1/2)T}I_{\mathrm{m}}(s)\dot{V}(s)\mathrm{d}s-\int_{(n+1/2)T}^{(n+1)T}I_{\mathrm{m}}(s)\dot{V}(s)\mathrm{d}s\right|\\
     =&\left|\int_{nT}^{(n+1/2)T}\sin(\omega s)\cos(\omega s)g(s)\mathrm{d}s\right.\\
     &\left.-\int_{(n+1/2)T}^{(n+1)T}\sin(\omega s)\cos(\omega s)g(s)\mathrm{d}s\right|.
\end{align*}
We take the absolute value to ensure that traversing a loop in either orientation yields the same $H$. In the Appendix Sec.~\ref{app:Hint} we show after a rather involved calculation that $H$ is given by
\begin{align}
	H=&\left|-\omega\sum_{k\in\mathds{O}}\alpha_{k}\left[\sum_{i=2}^{(k+1)/2}D_i(\omega)\frac{\Gamma\left(i\right)}{\Gamma\left(i-1/2\right)}\prod_{j=i+1}^{(k+1)/2}C_{2j-1}(\omega)\right.\right.\nonumber\\
 &\left.\left.+\frac{4}{3\left(1+\omega^2\right)}\prod_{n=1}^{(k-1)/2} C_{2n+1}(\omega)\right]\right|,\label{eq:H}
\end{align}
with $\mathds{O}$ the set of odd integers and
\begin{align*}
    C_j(\omega)=&\frac{\omega^2(j-1)j}{1+\omega^2j+\omega^2(j-1)j},\\
    D_i(\omega)=&\frac{4\sqrt{\pi}}{(1+2i)(1+(1-2i)^2\omega^2)}.
\end{align*}
We note that any contributions from possible $\beta_i$ terms are $T/2$ periodic and thus do not contribute to $H$ (as long as the SVM is of type I).

In general the $\omega$ for which $H$ is maximal depends on the values $\alpha_k$. If the function $h$ is known, then $H$ can be evaluated to observe for which value of $\omega$ the area is maximal, thereby uncovering the typical timescale of the device. In an experimental context, it is relatively straightforward to uncover the function $h$ since this is the steady-state conductance of the device, for which no dynamic measurements are required. In Sec.~\ref{sec:areaexample} we show that for the typical simple case of $\alpha_1\neq 0$ and $\alpha_{k}=0$ for all $k>1$, $k\in\mathds{O}$, that $H$ will always be maximal for $\omega=1$, regardless of the value of $\alpha_1$, where we recall that our dimensionless time $t$ is in units of the typical SVM memory retention time $\tau$. Therefore, in dimensionless time this means that $H$ is maximal if $\omega\tau=1$. We remark that $H$ is still maximal at $\omega\approx1$ for higher order terms in $h$, e.g.\ the $\alpha_3$ and $\alpha_5$ terms of $H$ are maximal at $\omega\approx0.91$ and $\omega\approx0.86$ respectively.
\begin{figure*}[ht]
	\centering
	\includegraphics[width=1\textwidth]{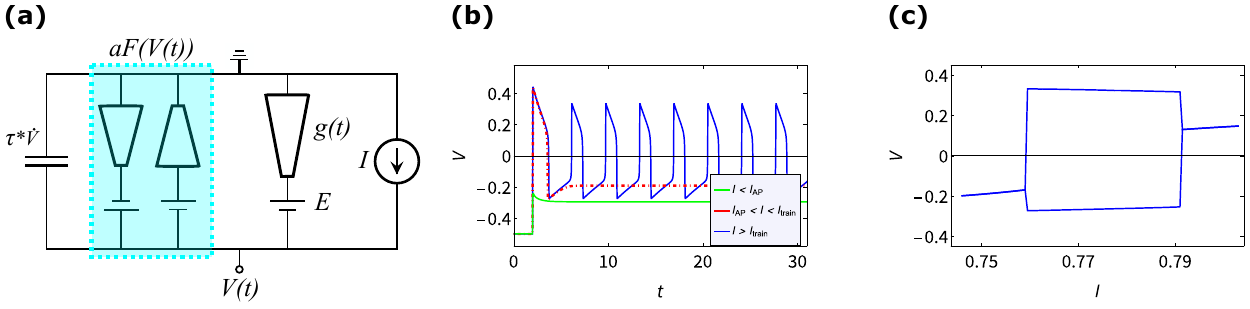}
	\caption{\textbf{(a)} Schematic circuit representation corresponding to Eqs.~(\ref{eq:mathCirc}) and (\ref{eq:mathCircg}). We distinguish four dimensionless current pathways, a capacitive current $\tau^*\dot{V}$, an imposed stimulus current $I$, a combined current  of  $aF(V(t))$ through two current rectifiers connected to individual batteries, and a current $g(t)(V(t)-E)$ through the SVM. These combined currents form Eq.~(\ref{eq:mathCirc}) by invoking Kirchhoff's law and as per Def.~\ref{def:SVMdef} $g(t)$ is described by Eq.~(\ref{eq:mathCircg}). Figure adapted from Ref.~\cite{kamsma2023iontronic}. \textbf{(b)} $V-t$ solutions for Eqs.~(\ref{eq:mathCirc}) and (\ref{eq:mathCircg}) for a subcritical input $I=0.65$ with no spike (green), a subcritical input that $I=0.75$ that generates a single spike (red) and a supercritical input $I=0.76$ that generated a periodic solution in the form of a spike train (blue). \textbf{(c)} A bifurcation diagram showing the Hopf bifurcations, with a spike train emerging for stimuli $I \in[0.7595,791]$.}
	\label{fig:deterministicCircuit}
\end{figure*}

\subsection{Comparison to fluidic microchannel memristor}\label{sec:areaexample}

As we showed in Secs.~\ref{sec:SVMexamp} and \ref{sec:comparisonIV}, the self-crossing hysteresis loop of a conical channel is well approximated by $h(x)=1+\alpha x$. Let us then calculate $H$ explicitly, where we recall that $V(t)=\sin(\omega t)$ with $\omega=2\pi/T$ is the dimensionless angular frequency and that our dimensionless time $t$ is in units of the typical SVM memory retention time $\tau$. A straightforward calculation shows that in this case
\begin{align*}
	H=-\frac{8\pi\alpha T}{12\pi^2+3T^2}=-\frac{(4/3)\alpha \omega}{\omega^2+1}.
\end{align*}
We can then investigate when $H$ is maximal,
\begin{gather*}
	\dfrac{\mathrm{d}H}{\mathrm{d}\omega}=-\frac{4\alpha (3 \omega^2+3)-24\alpha \omega^2}{(3 \omega^2+3)^2}=0\Rightarrow\omega=\pm 1.
\end{gather*}
Remarkably, the value of $\omega$ where $H$ is maximal in this case does not depend on $\alpha$. Therefore, any SVM with $h(x)=1+\alpha x$ and $V(t)=\sin(\omega t)$ will have a maximal enclosed area in the hysteresis loop at $\omega=1$. In dimensional units this means that $H$ is maximal if $\omega\tau=1$. This is again a general statement which should then also apply to a specific physical example of a conical channel, which we saw in Fig.~\ref{fig:conehourglass}(a) is well described by $h(x)=1+\alpha x$. In Fig.~\ref{fig:Hfig}(a) we show $H$ as a function of $\omega$, where we clearly see it indeed peaks at $\omega=1$. In Fig.~\ref{fig:Hfig}(b) we show the enclosed surface area in the $I-V$ diagram determined through physical continuum transport FE calculations of a conical channel.

To convert dimensional results form the FE calculations to the dimensionless $\omega$ used in Fig.~\ref{fig:Hfig}(b) we need the timescale $\tau$ that we used to convert to dimensionless time units. In Ref.~\cite{kamsma2023iontronic} an expression was derived for an estimate of $\tau$, however this is an approximation and not necessarily the precise value. To make the best comparison between the mathematical model and a specific realisation of an SVM, we estimate $\tau$ empirically. The steady-state conductance $h_{\infty}(V_{\mathrm{s}}(t))$ and equilibrium conductance $g_0$ of a conical channel are known (full details in the Appendix) \cite{boon2021nonlinear} so the only unknown for evaluating $\tau\dot{g}_{\mathrm{s}}(t)=g_0h_{\infty}(V_{\mathrm{s}}(t))-g_{\mathrm{s}}(t)$ to obtain dynamic conductance $g_{\mathrm{s}}(t)$ is the timescale $\tau$. Therefore we can treat $\tau$ as a fit parameter that yields the optimal solution $g_{\mathrm{s}}(t)$ to best fit the empirically found conductance with full FE calculations of the microscopic physical equations.

By comparing Figs.~\ref{fig:Hfig}(a) and \ref{fig:Hfig}(b) we indeed see a striking similarity. Both functions exhibit a nearly identical function form and the $H$ from the conical channel also peaks approximately at $\omega=1$ as predicted. Zooming in reveals that the peak in $H$ from the FE calculations is slightly offset to the left of $\omega=1$, which could indicate a more detailed agreement with Eq.~(\ref{eq:H}) as higher order terms in $h(x)$ slightly shift the peak of $H$ to $\omega<1$. Even though the conical channel is well described by a linear $h(x)$, there are in reality higher order terms present in the physical $h(x)$.

\section{SVM spiking circuit}\label{sec:circuit}
Thus far we have focused on the properties of a single SVM. Now we will turn our attention to an SVM based neuromorphic spiking circuit which was shown to exhibit key features of neuronal signalling \cite{kamsma2023iontronic}. The four dimensional physical circuit equations in Ref.~\cite{kamsma2023iontronic} can be reduced to a two-dimensional dynamical system featuring a vastly simplified parameter space, thereby offering the combination of a physically plausible system while remaining mathematically treatable. The emerging function forms bear a resemblance to that of the FitzHugh-Nagumo (FN) \cite{fitzhugh1961impulses,nagumo1962active} and Morris-Lecar (ML) \cite{morris1981voltage} models, but it is actually a different system as we will see. We first analyse the two-dimensional dynamical system that describes the circuit in Sec.~\ref{sec:detcircuit}. To make it explicit that our dynamical system corresponds to a physical system, we convert the physical governing system of equations \cite{kamsma2023iontronic} to its reduced dimensionless two-dimensional form in Sec.~\ref{sec:connection}. Then, in Sec.~\ref{sec:stoch} we show that the natural presence of voltage noise \cite{sarpeshkar1993white} can induce stochastic spiking, thereby presenting additional features of neuronal signaling \cite{teich1989fractal,shadlen1998variable,baddeley1997responses,nawrot2008measurement} from essentially the same circuit.
\begin{figure*}[ht]
	\centering
	\includegraphics[width=0.75 \textwidth]{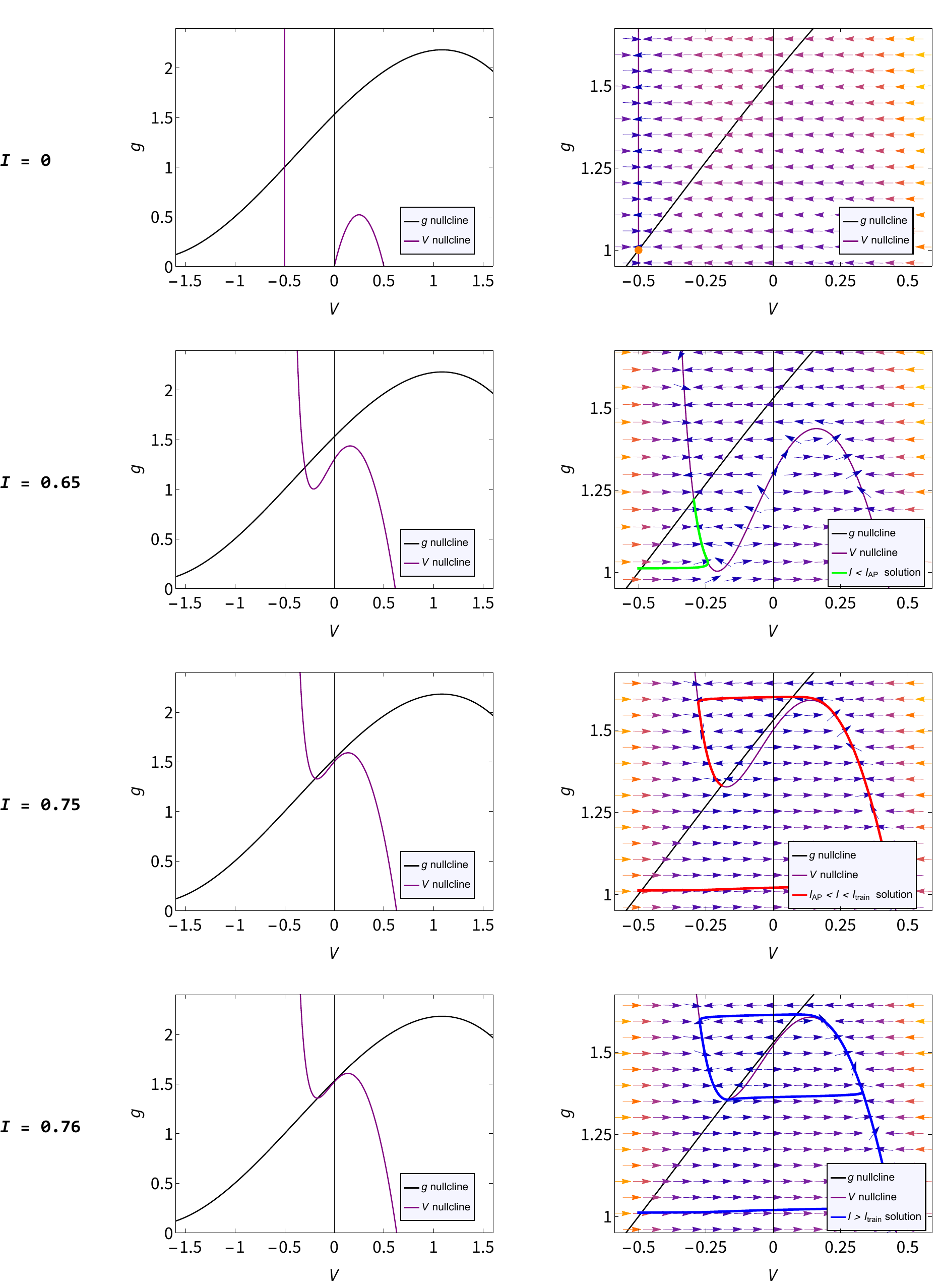}
	\caption{Phase portraits for Eqs.~(\ref{eq:mathCirc}) and (\ref{eq:mathCircg}) for no stimulus $I=0$, a subcritical stimulus $I=0.65$, a subcritical stimulus such that one spike occurs $I=0.75$, and a supercritical stimulus $I=0.76$. On the left we show the nullclines for $V$ (purple) and $g$ (black) for a wide regime. On the right we zoomed in on the physically relevant domain in the phase space, where we also show the vector field $(\dot{V},\dot{g})$ and the solutions that initiate from the steady-state when $I=0$ (indicated by the orange dot in the top-right figure).}
	\label{fig:phasePortraits}
\end{figure*}

\subsection{\emph{Deterministic} SVM spiking circuit}\label{sec:detcircuit}
Consider the circuit schematically illustrated in Fig.~\ref{fig:deterministicCircuit}(a) containing an SVM with conductance $g(t)$, over which a voltage $V(t)$ forms. In total there are four contributions to the dimensionless current in the circuit, given by i) a capacitive current $\tau^*\dot{V}$, ii) an imposed stimulus current $I$, iii) a combined current $aF(V(t))$ through two current rectifiers connected to individual batteries, and iv) a current $g(t)(V(t)-E)$ through the SVM. Therefore there are only two dynamical variables in this circuit, $V(t)$ and $g(t)$. With Kirchhoff's law we obtain Eq.~(\ref{eq:mathCirc}) for $V(t)$ and as per Def.~\ref{def:SVMdef}, we know that $g(t)$ described by Eq.~(\ref{eq:mathCircg}), yielding the following system of equations that describe the SVM spiking circuit

\begin{numcases}{}
    \tau^{*}\dfrac{\dd V(t)}{\dd t}=I-g(t)\left(V(t)-E\right)+aF(V(t)),\label{eq:mathCirc}\\
    \hspace{0.38 cm}\dfrac{\dd g(t)}{\dd t}= h(E-V(t))-g(t),\label{eq:mathCircg}
\end{numcases}
\begin{figure*}[ht]
	\centering
	\includegraphics[width=1\textwidth]{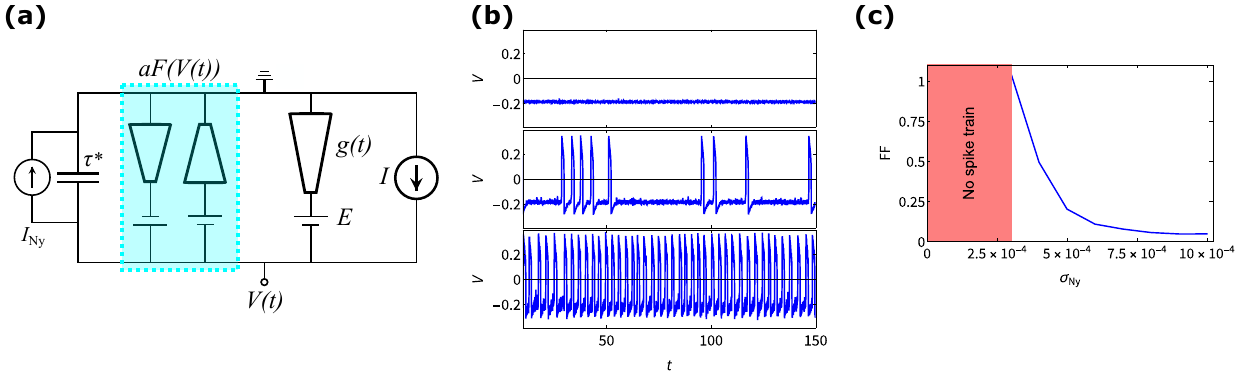}
	\caption{\textbf{(a)} Schematic representation of the SVM spiking circuit corresponding to Eqs.~(\ref{eq:dV}) and (\ref{eq:dg}). The noise is incorporated as an extra current $I_{\mathrm{Ny}}$ through the circuit. This figure is adapted from Ref.~\cite{kamsma2023iontronic}. \textbf{(b)} Spiking behaviour when the stimulus $I$ is subcritical (i.e.\ no spike train in deterministic case). When the noise is weak with $\sigma_{\mathrm{Ny}}=1\cdot10^{-4}$ the system exhibits no spikes (top), for intermediate noise with $\sigma_{\mathrm{Ny}}=4\cdot10^{-4}$ we see irregular spiking (middle) and for strong noise with $\sigma_{\mathrm{Ny}}=10\cdot10^{-4}$ the circuit exhibits a regular spike train (bottom). \textbf{(c)} The Fano Factor (FF) for various noise strengths, showing a regime where no spikes emerge (in the solved interval $t\in[0,10^3]$), a regime where spikes emerge irregularly resulting in a high FF, and a regime where we observe regular spiking which translates into a low FF.}
	\label{fig:StochCircuit}
\end{figure*}

where $F(x)=Gx-(Gx)^3/3$, $h(x)=1 - 1.07\,x + 0.06\,x^2 + 0.167\,x^3$ (as described in Sec.~\ref{sec:connection}), with $x\in\mathbb{R}$, and lastly $I\in\mathbb{R}$ is a control parameter, typically of order unity. Here we will use $\tau^{*}=0.01$, $E=-0.5$, $a=0.6$, and $G=3.46$. The system of Eqs.~(\ref{eq:mathCirc}) and (\ref{eq:mathCircg}) seems reminiscent of the FN and ML models \cite{fitzhugh1961impulses,nagumo1962active}, but the coupled $g(t)V(t)$ term is not present in the FN model. The ML model \cite{morris1981voltage} does feature such a coupled term, but lacks the simple cubic function $F(x)$. Therefore Eqs.~(\ref{eq:mathCirc}) and (\ref{eq:mathCircg}) represent a distinct model, which directly corresponds to a physical circuit as we show in Sec.~\ref{sec:connection}, and hence is physically fully plausible.

In Fig.~\ref{fig:deterministicCircuit}(b) we show three distinct solutions $V(t)$ for stimuli that undergo a step function from $I=0$ to $I=0.65$, $I=0.75$, and $I=0.76$, corresponding to no spike (green), a single isolated spike (red) and a spike train (blue), respectively. Thus we reproduce the earlier reported results \cite{kamsma2023iontronic} with Eqs.~(\ref{eq:mathCirc}) and (\ref{eq:mathCircg}). In Fig.~\ref{fig:deterministicCircuit}(c) the different regimes for different stimuli are clearly visible in a bifurcation diagram, showing a Hopf bifurcation at $I=I_{\mathrm{train}}\approx 0.7595$ and $I\approx 0.791$. The Jacobian of Eqs.~(\ref{eq:mathCirc}) and (\ref{eq:mathCircg}) is
\begin{align*}
    \mathbf{J}(V,g)=\begin{pmatrix}
        \frac{-g+aF^{\prime}(V)}{\tau^{*}}&-\frac{V-E}{\tau^{*}}\\
        -h^{\prime}(E-V)&-1
    \end{pmatrix}.
\end{align*}
Upon evaluating $\mathbf{J}(V,g)$ at the steady point(s) we see that indeed a Hopf bifurcation manifests at two distinct parameter values. The conjugate pair of eigenvalues crosses the imaginary axis for a stimulus $I\approx 0.7595$, rendering the real parts of the eigenvalues positive and marking the onset of oscillatory behaviour in the system around two new unstable steady points. A second Hopf bifurcation transpires at $I\approx 0.791$, where the real parts of the eigenvalues revert to negative values, indicating a return to a single stable steady point.

Eqs.~(\ref{eq:mathCirc}) and (\ref{eq:mathCircg}) feature the nullclines
\begin{numcases}{}
    f_{V}(V)=\frac{I+aF(V)}{V-E},\label{eq:vNull}\\
    f_{g}(V)= h(E-V),\label{eq:gNull}
\end{numcases}
where we denoted the $V$-nullcline  with $f_{V}$ and the $g$-nullcline with $f_{g}$. The difference with the FN model is particularly clear in Eq.~(\ref{eq:vNull}), which diverges around $V=E$ in our case due to the coupled $g(V-E)$ term in Eq.~(\ref{eq:mathCirc}). In Fig.~\ref{fig:phasePortraits} we show the nullclines for 4 different values of stimulus $I$, corresponding to no stimulus $I=0$, a subcritical stimulus $I=0.65$, a subcritical stimulus such that one spike occurs $I=0.75$, and a supercritical stimulus $I=0.76$. On the right hand side we zoom in on the physically relevant regime in the phase space, where we also show the vector field $(\dot{V},\dot{g})$ and the solutions that initiate from the steady-state when $I=0$ (indicated by the orange dot in the top-right figure). This reveals three distinct traces, one where the system directly settles to a new steady-state (green), one where the solution traverses the phase space once before settling (red) and a periodic solution (blue).  The green, red, and blue trajectories we show in Fig.~\ref{fig:phasePortraits} correspond to the green, red, and blue voltage traces we show in Fig.~\ref{fig:deterministicCircuit}(b).

\subsubsection{Physical fluidic spiking circuit equations}\label{sec:connection}

\begin{figure*}[ht]
	\centering
	\includegraphics[width=1\textwidth]{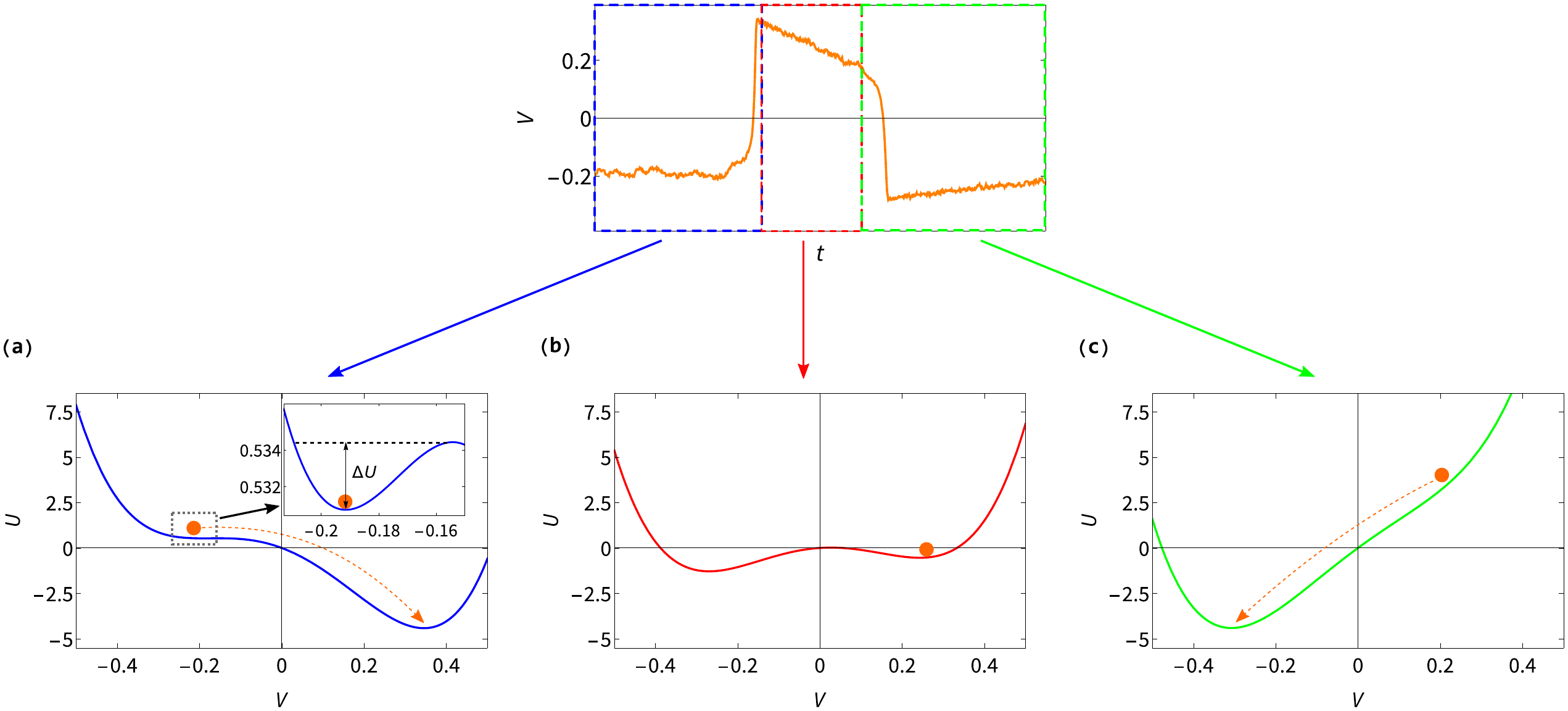}
	\caption{The potential $U^{*}(V,V_{\mathrm{g}})$ as in Eq.~(\ref{eq:Ustar}) for a subcritical stimulus $I=0.75$ for three different $V_{\mathrm{g}}$ corresponding to one (fluctuation induced) spike. \textbf{(a)} The quasi-stable state $V_{\mathrm{g}}=V^{*}$ case which features a local well of depth $\Delta U$ which can be spontaneously surpassed with a sufficiently strong fluctuation, driving $V$ to the deeper (temporary) right-side well. \textbf{(b)} An intermediate case $V_{\mathrm{g}}=0$ where $V_{\mathrm{g}}$ has increased from $V^{*}$ to 0, showing that the right-side well is disappearing while the left-side well has started appearing. \textbf{(c)} The right-side well completely disappears for $V_{\mathrm{g}}=V^{+}$ with $V^{+}$ the maximal $V$ during a spike, driving $V$ back to $V^{*}$.}
	\label{fig:PotentialPlots}
\end{figure*}

The system of Eqs.~(\ref{eq:mathCirc}) and (\ref{eq:mathCircg}) is directly derived from a theoretically proposed physical fluidic iontronic circuit \cite{kamsma2023iontronic}. The neuromorphic circuit in Ref.~\cite{kamsma2023iontronic} contains three conical ion channels (two with fast dynamics, one with slow), with $g_{\mathrm{s}}$ the conductance of the slow channel, and a capacitor over which a membrane potential $V_{\mathrm{m}}$ forms. This circuit was described by four differential equations, corresponding to the three channels and the potential $V_{\mathrm{m}}$. However, the two fast channels respond sufficiently quickly that they can be treated as instantaneous current rectifiers, rather than dynamical memristors. Consequently, the circuit dynamics were also shown to be well described by the following (dimensional) system of equations
\begin{numcases}{}
    \tau_{\mathrm{m}}\dfrac{\dd V_{\mathrm{m}}(t)}{\dd t}=\frac{I(t)}{g_{0}}-\frac{g_{\mathrm{s}}(t)}{g_{0}}\left(V_{\mathrm{m}}(t)-E_{\mathrm{s}}\right)+\frac{g_{\mathrm{r}}}{g_{0}}F(V_{\mathrm{m}}(t)),\label{eq:physCirc}\\
    \tau\dfrac{\dd g_{\mathrm{s}}(t)}{\dd t}= g_0h_{\infty}(-V_{\mathrm{m}}(t)+E_{\mathrm{s}})-g_{\mathrm{s}}(t),\label{eq:physchannel}
\end{numcases}
where Eq.~(\ref{eq:physCirc}) follows from Kirchhoff's law and describes the total current through the circuit and Eq.~(\ref{eq:physchannel}) describes the dynamical conductance of the slow channel in the circuit. Here, $I$ is a stimulus current, $E_{\mathrm{s}}=-0.5\text{ V}$ is a battery potential, and $g_{\mathrm{r}}=1\text{ pS}$ is a characteristic conductance scale. The membrane response RC-like time is denoted by $\tau_{\mathrm{m}}=C/g_{0}$ (in Ref.~\cite{kamsma2023iontronic} $\tau_{\mathrm{m}}$ is defined as $C/g_{\mathrm{r}}$, which is of similar magnitude). The function $F$ is the total current through the two fast channels given by
\begin{align}\label{eq:Fappr}
    F(V_{\mathrm{m}}(t))= GV_{\mathrm{m}}(t)-\frac{(GV_{\mathrm{m}}(t))^3}{3V_{\mathrm{r}}^2},
\end{align}
where $V_{\mathrm{r}}=1\text{ V}$ is a characteristic voltage scale and $G=3.46$

Let us then rewrite these equations such that they become dimensionless. We define $t\to t/\tau$, $\tau^{*}=\tau_{\mathrm{m}}/\tau$, $V_{\mathrm{m}}(t)/V_{\mathrm{r}}\to V(t)$, $I(t)/(g_{0}V_{\mathrm{r}})\to I$, $E_{\mathrm{s}}/V_{\mathrm{r}}\to E$, $g=g_{\mathrm{s}}(t)/g_{\mathrm{0}}$, and $a=g_{\mathrm{r}}/g_0$. Additionally, to be in agreement with Eq.~\ref{eq:h} we define $h(E-V(t))$ in this Section as the third order expansion of $h_{\infty}(V_{\mathrm{r}}(E-V(t)))$ given in Eq.~(\ref{eq:rhos}), i.e.\
\begin{align*}
    h(V(t))=&h_{\infty}(V_{\mathrm{r}}E)+\sum_{i=1}^{3}\frac{1}{i!}\left.\dfrac{\dd h_{\infty}(V_{\mathrm{r}}(E-V))}{\dd V}\right|_{V=0}(V(t))^i\\
    =&h_{\infty}(V_{\mathrm{r}}E)+\sum_{i=1}^{3}\frac{V_{\mathrm{r}}(-1)^i}{i!}\left.\dfrac{\dd h_{\infty}(\nu)}{\dd \nu}\right|_{\nu=E}(E-V(t))^i\\
    \approx&1 - 1.07\,V(t) + 0.06\,(V(t))^2 + 0.167\,(V(t))^3.
\end{align*}
where we used $\nu=V_{\mathrm{r}}(E-V)$. This yields the function $h$ used in Eq.~(\ref{eq:mathCircg})
\begin{align}\label{eq:hg}
    h(x)=1 - 1.07\,x + 0.06\,x^2 + 0.167\,x^3.
\end{align}
With this we arrive at the dimensionless Eqs.~(\ref{eq:mathCirc}) and (\ref{eq:mathCircg}) and effectively simplify the originally extensive parameter space to a simple two dimensional system that contains only 4 parameters. 

\subsection{\emph{Stochastic} SVM spiking circuit}\label{sec:stoch}
Thermal noise is an intrinsic element of RC circuits \cite{sarpeshkar1993white}, often referred to as Nyquist noise \cite{johnson1928thermal,nyquist1928thermal}. Since we are essentially working with an RC-circuit, there will be thermal fluctuations in the circuit voltage of order \cite{sarpeshkar1993white} $\sqrt{k_{\mathrm{B}}T/C}\approx 3.7 \text{mV}$, with $k_{\mathrm{B}}$ the Boltzmann constant, $T=293.15$ K room temperature, and $C=0.3$ fF the capacitance. A natural way to model the presence of voltage noise is by adding Brownian fluctuations to the ODE (\ref{eq:mathCirc}).  As we shall show, the solutions of the corresponding Stochastic Differential Equations (SDEs) present interesting modes of neuronal spiking, not found in the deterministic case.
\begin{figure}[ht]
	\centering
	\includegraphics[width=0.5\textwidth]{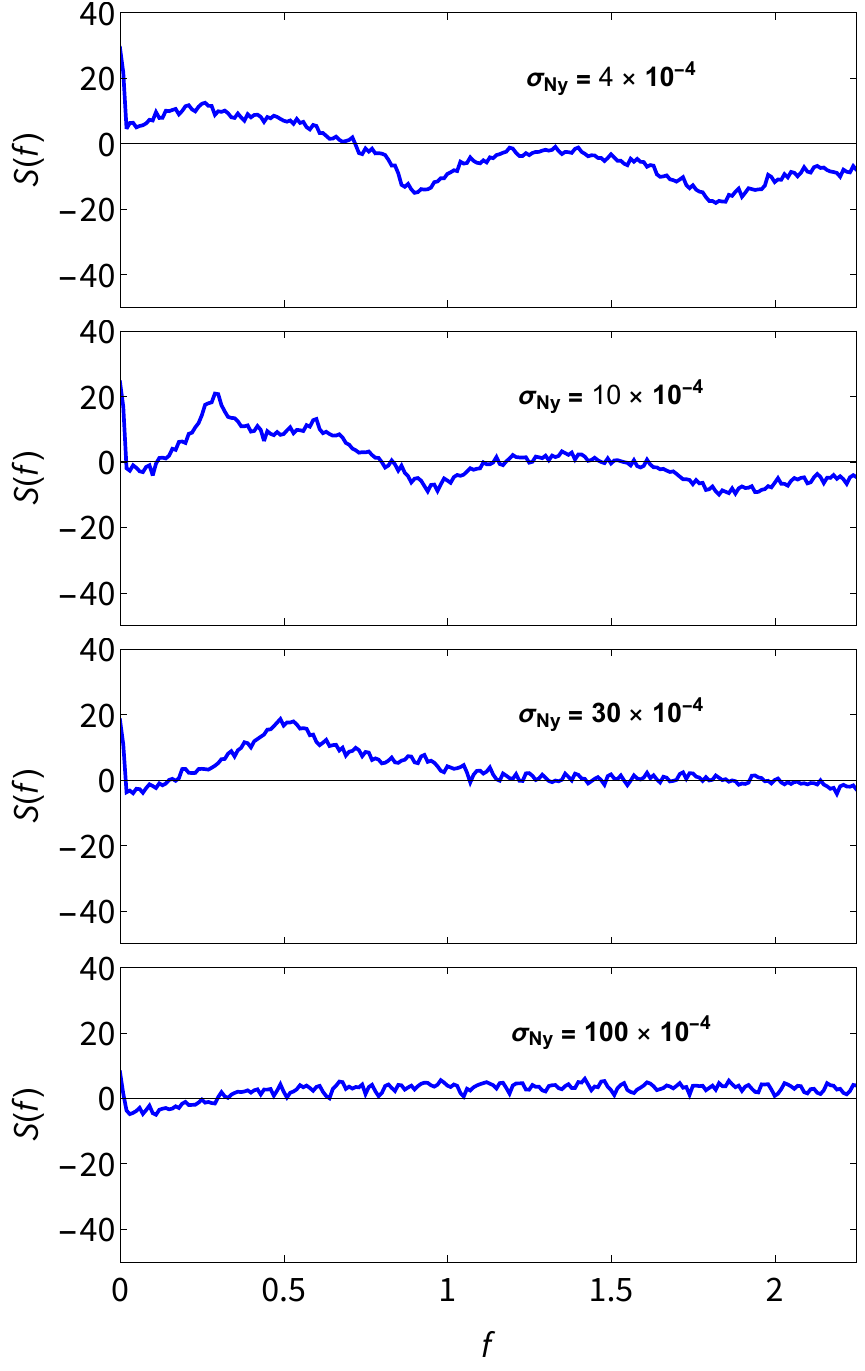}
	\caption{Power spectra $S(f)$, with $f$ the frequency, of the solutions generated by the SDEs (\ref{eq:dV}) and (\ref{eq:dg}) for various noise strengths $\sigma_{\text{Ny}}$. For $\sigma_{\text{Ny}}=4\cdot10^{-4}$ we find sporadic spiking, resulting in a spread out power spectrum. For $\sigma_{\text{Ny}}=10\cdot10^{-4}$ and $\sigma_{\text{Ny}}=30\cdot10^{-4}$ we find more regular spiking, visible as a clear spike in the power spectrum. Lastly, for $\sigma_{\text{Ny}}=100\cdot10^{-4}$ the system is dominated by noise and consequently no clear peak emerges anymore.}
	\label{fig:PSPlots}
\end{figure}

We incorporate the noise by the replacement $V(t)\to V(t)+\sigma_{\mathrm{Ny}} \dd W_{\mathrm{t}}$ in Eq.~(\ref{eq:mathCirc}), which would result in an extra noise current of magnitude $\dd I_{\mathrm{Ny}}=(g+h_{+}(V)+h_{-}(V))\sigma_{\mathrm{Ny}}\dd W_{\mathrm{t}}/\sqrt{\tau^{*}}=\epsilon(V)\dd W_{\mathrm{t}}$, where $h_{\pm}(V)$ is the instantaneous conductance of the two channels that form the current $aF(V(t))$, approximated by a third order expansion of the physical functions from Ref.~\cite{kamsma2023iontronic}, laid out in the Appendix. Additionally, we can write Eq.~(\ref{eq:mathCirc}) in terms of a time-dependent potential $U^{\prime}(V(t))=-\frac{1}{\tau^{*}}\left[I-g(t)\left(V(t)-E\right)+aF(V(t))\right]$. This yields the following compact stochastic differential equations
\begin{numcases}{}
    \dd V(t)=-U^{\prime}(V(t))\dd t+\epsilon(V(t))\dd W_{\mathrm{t}},\label{eq:dV}\\
    \dd g(t)= \left[h(E-V)-g(t)\right]\dd t,\label{eq:dg}
\end{numcases}
where $\epsilon(V):=\frac{1}{\sqrt{\tau^{*}}}(g+h_{+}(V)+h_{-}(V))\sigma_{\mathrm{Ny}}$.

If we consider voltage noise in the circuit, we would obtain the equivalent circuit schematically depicted in Fig.~\ref{fig:StochCircuit}(a), which is the circuit we showed in Fig.~\ref{fig:deterministicCircuit}(a) adapted to include a noise current component $I_{\mathrm{Ny}}$. In the deterministic limit $\sigma_{\mathrm{Ny}}\to0$ we then see that the noise current $I_{\mathrm{Ny}}\to0$ such that we recover the original circuit and the results presented in Fig.~\ref{fig:deterministicCircuit}. In Fig.~\ref{fig:StochCircuit}(b) we show solutions for Eqs.~(\ref{eq:dV}) and (\ref{eq:dg}) solved on the interval $t\in[0,10^3]$ with $I=0.75<I_{\mathrm{train}}$ (i.e.\ subcritical) for weak $\sigma_{\mathrm{Ny}}=1\cdot10^{-4}$, intermediate $\sigma_{\mathrm{Ny}}=4\cdot10^{-4}$, and strong noise $\sigma_{\mathrm{Ny}}=10\cdot10^{-4}$ in the top, middle, and bottom figure, respectively. For weak noise the system essentially shows the same subcritical behaviour as the deterministic case, i.e.\ Fig.~\ref{fig:deterministicCircuit}(b,green), exhibiting no spikes in the interva $t\in[0,10^3]$. When we move on to intermediate noise then we observe a new response in the form of irregular spiking, rather than the periodic ``tonic'' firing that we observe in Fig.~\ref{fig:deterministicCircuit}(b,blue). Then for strong noise we retrieve essentially regular tonic spiking, even though the stimulus is subcritical. 

By reviewing Fig.~\ref{fig:StochCircuit}(b) it is clear that various forms of noise induce various forms of spiking with differing regularities. The Fano Factor (FF) is a measure of the variability of a stochastic process, which in the context of neuronal spiking is given by the ratio of the variance and the squared expectation value of the Inter-Spike Interval ISI (the time between subsequent spikes) \cite{dayan2005theoreticalFF}, thus $\text{FF}=\mathrm{Var}(\mathrm{ISI})/(\mathrm{E}\left[\mathrm{ISI}\right])^2$. When we look at the corresponding FF for various strengths of noise we obtain Fig.~\ref{fig:StochCircuit}(c). Initially the noise is not strong enough to induce spiking in the interval $t\in[0,10^3]$, after which sporadic spikes or bursts occur resulting in a high FF. Lastly, for stronger noise the essentially tonic spiking results in a vanishing FF. In biological neuronal spiking it is common for neurons to exhibit a non-zero FF \cite{teich1989fractal,shadlen1998variable,baddeley1997responses,nawrot2008measurement}, therefore the inclusion of noise effectively captures another characteristic of neuronal signalling.

We can make sense of the stochastic features we observe in Fig.~\ref{fig:StochCircuit} by analysing the potential $U$ in Eq.~(\ref{eq:dV}). Recall that $\tau^{*}\ll 1$, the resulting fast-slow characteristic of Eqs.~(\ref{eq:dV}) and (\ref{eq:dg}) ensures that $V$ exhibits much faster dynamics than $g$. This is especially visible in Fig.~\ref{fig:phasePortraits}, where we see that the vectors $(\dot{V},\dot{g})$ are essentially horizontal and in the solutions we see that $g$ remains essentially constant when $V$ changes from a negative to positive state (or vice versa). We can use this separation in timescale to simplify the coupled relation of Eqs.~(\ref{eq:dV}) and (\ref{eq:dg}) in this context by considering $g$ to be equal to its steady-state function $h(V_{\mathrm{g}})$ for a given $V_{\mathrm{g}}$, yielding the following equation
\begin{align}\label{eq:Ustar}
    U^{*}(V,V_{\mathrm{g}})=-\frac{1}{\tau^{*}}\int\left[I-h(V_{\mathrm{g}})\left(V-E\right)+aF(V)\right]\dd V.
\end{align}
We will consider $U^{*}$ in the subcritical case $I=0.75$ used for Fig.~\ref{fig:StochCircuit}(b), such that the deterministic case is stable at $V^{*}\approx-0.187$. Without any noise the system will remain in its corresponding stable state $(V^{*}, h(V^{*}))$. The potential $U^{*}(V^{*},V^{*})$ is shown in Fig.~\ref{fig:PotentialPlots}(a) where we see that $(V^{*}, h(V^{*}))$ is a shallow stable steady state in an energy well of depth $\Delta U\approx3.7\cdot10^{-3}$. Once spikes start appear in the solved time interval, we indeed find a sufficiently short Kramers' time. E.g.\, the Kramers' time $T_\mathrm{Kramers}=e^{2\Delta U/\epsilon(V^{\star})}\approx 1.1$ for the intermediate noise strength used for Fig.~\ref{fig:StochCircuit}(b, middle). 

If a spike is induced then $V$ increases to $V^{+}$ quickly compared to $g$ since $\tau^{*}\ll 1$. Consequently there is a period where the system (approximately) is in the state $(V^{+},h(V^{*}))$ (which is the right side deeper well in Fig.~\ref{fig:PotentialPlots}(a)), and not in $(V^{+},h(V^{+}))$. This features translates into the all-or-none behaviour of the action potentials. Either the stimulus is not strong enough to exceed the energy barrier $\Delta U$, i.e.\ no spike occurs, or the system completely transitions to the second deeper well, i.e.\ a complete spike occurs. Then $V_{\mathrm{g}}$ will transition from $V^{*}$ to $V^{+}$ over a timescale of order $\mathcal{O}(1)$ (or $\mathcal{O}(\tau)$ in dimensional units), i.e.\ the typical timescale of the dynamics of $g$. In Fig.~\ref{fig:PotentialPlots}(b) we show an intermediate case where $V_{\mathrm{g}}=0$, the right-side well still exists, so $V$ remains in its positive state, but the left-side well has re-emerged. In Fig.~\ref{fig:PotentialPlots}(c) we plot $U^{*}$ for $V_{\mathrm{g}}=V^{+}$, so the case when $g$ has had time to transition to the state $h(V^{+})$. The right-side well has now completely disappeared and the system transitions back to the left-side well, eventually returning $V$ to its subcritical steady state $V^{*}$.

The results shown in Fig.~\ref{fig:PotentialPlots} explain the observed FF results shown in Figs.~\ref{fig:StochCircuit}(b) and \ref{fig:StochCircuit}(c). For a weak noise, the fluctuations are not strong enough to escape the shallow energy well of depth $\Delta U$ in the interval $t\in[0,10^3]$, so no spikes are found. For intermediate noise, the shallow energy well can be escaped, but only sporadically as the fluctuations are only rarely strong enough, resulting in irregular spiking with a high FF. Lastly for strong noise the fluctuations are strong enough to essentially always escape the shallow well whenever the system is in the state depicted in Fig.~\ref{fig:PotentialPlots}(a), thus quickly generating a new spike after the previous ended. This results into essentially regular spiking, and a low FF, despite the subcritical stimulus.

In Fig.~\ref{fig:PSPlots}, we present the power spectra derived from the solutions of the SDEs (\ref{eq:dV}) and (\ref{eq:dg}), displaying the influence of noise strength on the spectral characteristics of the system dynamics. Under weak noise $\sigma_{\text{Ny}}=4\cdot10^{-4}$, the power spectrum exhibits a spread out peak, reflecting the sporadic and irregular nature of spiking events. As the noise intensity increases to $\sigma_{\text{Ny}}=10\cdot10^{-4}$ and $\sigma_{\text{Ny}}=30\cdot10^{-4}$, a discernible sharper peak emerges in the power spectrum, indicative of a more structured periodicity in the system behaviour akin to the tonic firing of the deterministic system. Lastly, for very strong noise levels $\sigma_{\text{Ny}}=100\cdot10^{-4}$, the power spectrum shows no particular peaks due to the dominating influence of noise on the system dynamics.

\section{Conclusion}\label{sec:conclusion}
In summary, we established a concise and simple mathematical framework for Simple Volatile Memristors (SVMs), catering to a broad class of volatile memristors. We derive this framework, specifically Def.~\ref{def:SVMdef}, with a general volatile memristor in mind. Of particular interest is a novel family of fluidic memristors that use ions in an aqueous environment as signal carriers, inspired by the brain's aqueous ion transport \cite{yu2023bioinspired,khan2023advancement,hou2023learning,kim2023liquid,xie2022perspective}. Several of these fluidic devices, as well as naturally occurring memristors in plants \cite{markin2014analytical}, were found to be well described by our definition of SVMs \cite{robin2023long,kamsma2023iontronic,kamsma2024brain}, therefore the mathematical framework discussed here is of direct relevance to novel physical systems. Throughout Secs.~\ref{sec:typeIandII} and \ref{sec:area} we showed that our mathematical predictions materialised in full FE calculations of microscopic physical continuum transport equations that incorporate no prior knowledge of our mathematical definitions. Specifically, in Sec.~\ref{sec:typeIandII} we proved that symmetric SVMs will produce a current-voltage hysteresis loop that does not cross itself in the origin, while an asymmetric SVM (under reasonable constraints) will produce self-crossing hysteresis loops. These predictions are then demonstrated through FE calculations of asymmetric (conical) and symmetric (hourglass shaped) microfluidic ion channels which indeed produce self-crossing and non self-crossing hysteresis loops, respectively. In Sec.~\ref{sec:area} we presented a general equation for the enclosed area within one self-crossing loop and explain how this can be used to estimate the typical conductance memory timescale $\tau$ of an SVM. We show that the maximal enclosed area is found when $2\pi f\tau\approx1$, with $f$ the (dimensional) frequency of the applied (sinusoidal) voltage and $\tau$ the memory timescale of the device. For SVMs where the steady-state conductance is linear in the voltage, the relation $2\pi f\tau=1$ is exact. This is again shown to also materialise within physical continuum transport FE calculations of a conical microchannel, showing good agreement with the general mathematical predictions. 

Building on our SVM framework, we turned our attention to a proposed application of SVMs, a circuit that exhibits key features of neuronal signaling \cite{kamsma2023iontronic}, initially described by a coupled four-dimensional system of physical equations relying on a large parameter space. We first present this circuit as a two-dimensional system, with variables $g$ and $V$, of equations relying on merely 4 parameters, which we explicitly show to be completely equivalent the original physical equations. This mathematically treatable system is first characterised by investigating its phase portraits revealing that it is in fact rather different from the seemingly similar FitzHugh-Nagumo \cite{fitzhugh1961impulses,nagumo1962active} and Morris-Lecar models \cite{morris1981voltage}. We then show that we can use the intrinsic instability of the system in combination with stochastic effects, naturally present in any circuit \cite{sarpeshkar1993white}, to extract features of irregular spiking not found in the deterministic case. These features are well understood through a mathematical analysis of the underlying equations, revealing a double-well energy potential structure that requires sufficiently strong noise to spontaneously escape a local minimum, thereby inducing a spike. Our work here provides insights into a physically plausible description of a neuromorphic spiking system, thus offering the combination of a mathematically treatable two-dimensional system of equations that has an explicit and direct correspondence to physical equations. 

For future prospects, the spiking circuit could be expanded to feature more SVMs. An expanded (physical) SVM circuit has already been proposed that features additional bursting spiking modes \cite{kamsma2024advanced}, which could undergo a similar mathematical reduction to a (in this case three-dimensional) dynamical system. Additionally, the simple mathematical model that described the SVM circuit paves the way for mathematically modeling networks of SVM based spiking ``neurons". Lastly, the emergence of spiking in SVM based circuits is sensitive to parameter changes \cite{kamsma2023iontronic,kamsma2023unveiling,kamsma2024advanced}. The mathematical model we provide here could be used to find more stable conditions for spiking which can then (possibly) be reverse engineered to be featured in actual devices, where the clear and explicit link we provide between our mathematical model and physical quantities can be of help.

In conclusion, by defining a relatively simple mathematical framework we are able to make some remarkable powerful statements on a wider class of memristive devices. We make explicit that our work pertinent to actual physical systems by explicitly transforming the system parameters to the physical ones that describe fluidic microchannels and by showing that our general mathematical predictions materialised in physical simulations on multiple occasions. This places our work at an interesting intersection between mathematics and physics in the context of the emerging field of iontronic neuromorphics, where insights are gained through mathematical analysis while retaining a clear and explicit connection to specific and relevant physical systems, devices, and applications. Our results here provide general tools for the characterisation of SVMs and spiking circuit applications thereof, specifically in the quickly emerging fascinating direction of fluidic iontronic devices.

\section*{Declaration of competing interest}
The authors declare that they have no known competing financial interests or personal relationships that could have appeared to influence the work reported in this paper.

\section*{Data availability}
All data required to convey scientific findings described in the paper are present in the main text or the Appendix.

\begin{acknowledgments}
We thank Dr.\ Jaehyun Kim and Prof.\ Jungyul Park for providing us with Fig.~\ref{fig:IntroFig}(c).
\end{acknowledgments}

\clearpage
\onecolumngrid
\appendix
\section{Appendix}
\subsection{Poisson-Nernst-Planck-Stokes equations}\label{sec:PNP}

In order to demonstrate the realisation of our mathematical predictions in Secs.~\ref{sec:typeIandII} and \ref{sec:area}, we conducted finite-element (FE) calculations using the FE analysis package COMSOL \cite{multiphysics1998introduction,pryor2009multiphysics}. To do so, we consider an azimuthally symmetric single conical channel, schematically depicted in Fig.~\ref{fig:cone}, of length $L=10\text{ }\mu\text{m}$ with the central axis at radial coordinate $r=0$ and a radius described by $R(x)=R_{\mathrm{b}}-x\Delta R/L$ for $x\in\left[0,L\right]$ where $R_{\mathrm{b}}=200\text{ nm}$ at $x=0$ and $R_{\mathrm{t}}=R_{\mathrm{b}}-\Delta R=50\text{ nm}$ at $x=L\gg R_{\mathrm{b}}$. For the hourglass channel results, the geometry we constructed consists of two conical channels of length $L=5\text{ }\mu\text{m}$ connected by the tip, with all other parameters identical. The channel connects two bulk reservoirs of an incompressible aqueous 1:1 electrolyte with viscosity $\eta=1.01\text{ mPa}\cdot\text{s}$, mass density $\rho_{\mathrm{m}}=10^3\text{ kg}\cdot\text{m}^{-3}$ and electric permittivity $\epsilon=0.71\text{ nF}\cdot\text{m}^{-1}$, at the far side of both reservoirs we impose a fixed pressure $P=P_0$ and fixed ion concentrations $\rho_{\pm}=\rho_{\mathrm{b}}=1\text{ mM}$. The channel wall carries a uniform surface charge density $e\sigma=-0.02\;e\text{nm}^{-2}$, screened by an electric double layer with Debye length $\lambda_{\mathrm{D}}\approx10\text{ nm}$, that imposes an electric surface potential of $\psi_0\approx -40\text{ mV}$. The ions have diffusion coefficients $D_{\pm}=D=1\text{ }\mu\text{m}^2\text{ms}^{-1}$ and charge $\pm e$ with $e$ the proton charge. On the far side of the reservoir connected to the base we impose an electric potential $V(t)$, while the far side of the other reservoir is grounded, which leads to an electric potential profile $\Psi(x,r,t)$, an electro-osmotic fluid flow with velocity field $\mathbf{u}(x,r,t)$ and ionic concentrations and fluxes $\rho_\pm(x,r,t)$ and $\mathbf{j}_{\pm}(x,r,t)$, respectively.

The transport dynamics are assumed to be described by the Poisson-Nernst-Planck-Stokes (PNPS) Eqs.~(\ref{eq:poisson})-(\ref{eq:stokes}) given by 
\begin{gather}
	\nabla^2\Psi=-\frac{e}{\epsilon}(\rho_+-\rho_-),\label{eq:poisson}\\
	\dfrac{\partial\rho_{\pm}}{\partial t}+\nabla\cdot\mathbf{j}_{\pm}=0,\label{eq:ce}\\
 \mathbf{j}_{\pm}=-D_{\pm}\left(\nabla\rho_{\pm}\pm\rho_{\pm}\frac{e\nabla \Psi}{k_{\mathrm{B}}T}\right)+\mathbf{u}\rho_{\pm},\label{eq:NP}\\
	\rho_{\mathrm{m}}\dfrac{\partial\mathbf{u}}{\partial t}=\eta\nabla^2\mathbf{u}-\nabla P-e(\rho_+-\rho_-)\nabla \Psi;\qquad\nabla\cdot\mathbf{u}=0.\label{eq:stokes}
\end{gather}
The Poisson Eq.~(\ref{eq:poisson}) accounts for electrostatics, the continuity Eq.~(\ref{eq:ce}) ensures the conservation of ions, the Nernst-Planck Eq.~(\ref{eq:NP}) incorporates the combination of Fickian diffusion, Ohmic conduction, and Stokesian advection, and finally the Stokes Eq.~(\ref{eq:stokes}) describes the force balance on the (incompressible) fluid. To make the system of Eqs.~(\ref{eq:poisson})-(\ref{eq:stokes}) closed we impose  boundary conditions of no-slip blocking flux and Gauss' law on the channel wall, i.e.\ $\mathbf{u}=0$, $\mathbf{n}\cdot\mathbf{j}_{\pm}=0$, and $\mathbf{n}\cdot\nabla\Psi=-e\sigma/\epsilon$ with $\mathbf{n}$ the wall's inward normal vector.
\begin{figure}[h]
	\centering
	\includegraphics[width=0.42\textwidth]{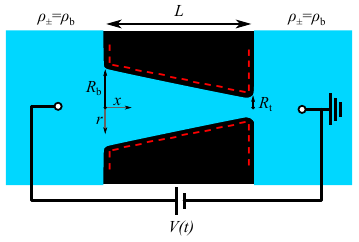}
	\caption{Schematic representation of an azimuthally symmetric conical channel of length $L$, base radius $R_{\mathrm{b}}$, and tip radius $R_{\mathrm{t}}<R_{\mathrm{b}}$, connecting two bulk reservoirs of a 1:1 aqueous electrolyte, with bulk concentrations $\rho_{\mathrm{b}}$. The channel wall carries a surface charge density $e\sigma$. An AC electric potential drop $V(t)$ over the channel drives an ionic charge current $I(t)=g_{\mathrm{s}}(t)V(t)$ with $g(V(t),t)$ the channel conductance. Figure adapted from Ref.~\cite{kamsma2023iontronic}.}
	\label{fig:cone}
\end{figure}
\subsection{Channel conductance}
The steady state conductance of a conical channel depends on the voltage-dependent salt concentration profile $\overline{\rho}_{\mathrm{s}}(x,V)$, that exhibits salt concentration polarisation upon an applied voltage. The consequent voltage-dependent steady-state channel conductance is described by \cite{boon2021nonlinear}
\begin{equation}\label{eq:rhos}
    \begin{aligned}
    h_{\infty}(V)=\frac{g_{\infty}(V)}{g_0}=&\int_0^L \overline{\rho}_{\mathrm{s}}(x,V) \dd x/(2\rho_{\rm{b}}L)\\
    =&1+\Delta g\int_0^{L}\left[\frac{x}{L}\frac{R_{\mathrm{t}}}{R(x)}-\frac{e^{\text{Pe}(V)\frac{x}{L}\frac{ R_{\mathrm{t}}^2}{R_{\mathrm{b}}R(x)}}-1}{e^{\text{Pe}(V)\frac{R_{\mathrm{t}}}{R_{\mathrm{b}}}}-1}\right]\dd x/L,
    \end{aligned}
\end{equation}
where $g_0=(\pi R_{\mathrm{t}} R_{\mathrm{b}}/L)(2\rho_{\rm{b}}e^2D/k_{\mathrm{B}}T)$, $\text{Pe}(V)=Q(V)L/(D\pi R_{\mathrm{t}}^2)$ the P\'{e}clet number at the narrow end, $Q(V)=-\pi R_{\mathrm{t}}R_{\mathrm{b}}\epsilon\psi_0/(\eta L)V$ the volumetric fluid flow  through the channel, and $\Delta g\equiv-e\Delta R\eta\sigma D/(\rho_{\mathrm{b}}R_{\mathrm{b}}R_{\mathrm{t}}\epsilon\psi_0k_{\mathrm{B}}T)$.

The dynamic (dimensional) conductance $g_{\mathrm{s}}(t)$ was found to be well described by \cite{kamsma2023iontronic}
\begin{align}
    \dfrac{\dd g_{\mathrm{s}}(t)}{\dd t}=\frac{g_0h_{\infty}(V(t))-g_{\mathrm{s}}(t)}{\tau},
\end{align}
with $\tau=L^2/12D$ the typical conductance memory time of the channel.

The circuit in Ref.~\cite{kamsma2023iontronic} contains three channels with a voltage dependent conductance. The ``fast" channels can be assumed to be instantaneous and are therefore not treated as dynamical equations in the circuit in this work, simplifying the system of equations to a two-dimensional one. To model the noise current we use the steady-state conductance of these channels, which we approximate by a third-order expansion of the physical equations from Ref.~\cite{kamsma2023iontronic}, this gives
\begin{align*}
    h_{+}(x)=&h_{+,\infty}(V_{\mathrm{r}}E_{\mathrm{+}})+\sum_{i=1}^{3}\frac{1}{i!}\left.\dfrac{\dd h_{+,\infty}(V_{\mathrm{r}}(E_+-V))}{\dd V}\right|_{V=0}x^i\approx 7.89+9.16x+6.76x^2+2.91x^3,\\
    h_{-}(x)=&h_{-,\infty}(V_{\mathrm{r}}E_{\mathrm{-}})+\sum_{i=1}^{3}\frac{1}{i!}\left.\dfrac{\dd h_{-,\infty}(V_{\mathrm{r}}(V-E_{-}))}{\dd V}\right|_{V=0}x^i\approx 7.89-9.16x+6.76x^2-2.91x^3.
\end{align*}
Note that these channels are normalized by $g_0$, which is the Ohmic conductance of the third ``slow" channel. Therefore $h_{\pm}(0)\neq 1$, unlike how we defined SVMs in Def.~\ref{def:SVMdef}. This is not an issue since these channels are not considered to be SVMs in the treatment of the circuit here, but rather as (instantaneous) current-rectifying circuit elements.

\subsection{Enclosed surface area $H$}\label{app:Hint}
The enclosed area $H$ inside a hysteresis loops, resulting rom a sweeping potential $V(t)$, that crosses itself in the origin and does not intersect itself anywhere else is given by
\begin{align*}
	H=&\left|\int_{nT}^{(n+1/2)T}I_{\mathrm{m}}(s)V'(s)\mathrm{d}s-\int_{(n+1/2)T}^{(n+1)T}I_{\mathrm{m}}(s)V'(s)\mathrm{d}s\right|\\
     =&\left|\int_{nT}^{(n+1/2)T}g(s)V(s)V'(s)\mathrm{d}s-\int_{(n+1/2)T}^{(n+1)T}g(s)V(s)V'(s)\mathrm{d}s\right|.
\end{align*}
The absolute value bars are there to account for either orientation of the hysteresis loop. We will consider a typical sinusoidal sweeping potential $V(t)=\sin(\omega t)$ and we focus on the first integral for simplicity, but all calculations below apply similarly to the second integral. Using the general solution Eq.~(\ref{eq:intsol}) we find
\begin{align*}
    \int_{nT}^{(n+1/2)T}g(s)V(s)V'(s)\mathrm{d}s=&\omega\int_{nT}^{(n+1/2)T}\cos(\omega s)\sin(\omega s)e^{-s}\left[\int_{0}^{s}h(V(t^{\prime}))e^{t^{\prime}}\mathrm{d}t^{\prime}+g(0)\right]\mathrm{d}s.
\end{align*}
From here on out we will ignore the $g(0)$ term as this represents a transient that vanishes due to the $e^{-s}$ term. Moreover, as discussed in  the main text we can choose $g(0)$ such that the solution for $g(t)$ is periodic without any transients. Therefore we can safely ignore transient terms here and that appear later on, as these cancel out with this choice of $g(0)$ that facilitates a periodic solution. Additionally we can ignore any $\beta_j$ terms in $h(V(t))$ as these even terms terms result in equal area terms of the two loop components. However, as we assume a self-crossing hysteresis loop, we subtract the second integral from the first integral due to the opposite orientations of the two loop components. Therefore the equal area terms cancel each other out and do not contribute to $H$. With these considerations we calculate
\begin{align*}
    \int_{nT}^{(n+1/2)T}g(s)V(s)V'(s)\mathrm{d}s=&\omega\sum_{k\in\mathds{O}}\alpha_{k}\int_{nT}^{(n+1/2)T}\cos(\omega s)\sin(\omega s)e^{-s}\int_{0}^{s}\sin(\omega t^{\prime})^{k}e^{t^{\prime}}\mathrm{d}t^{\prime}\mathrm{d}s\\
    =&\omega\sum_{k\in\mathds{O}}\alpha_{k}\int_{nT}^{(n+1/2)T}\cos(\omega s)\sin(\omega s)A_{k}\mathrm{d}s.
\end{align*}
Here $\mathds{O}$ is the set of odd integers and we denoted $A_{k}=e^{-s}\int_{0}^{s}\sin(\omega t^{\prime})^{k}e^{t^{\prime}}\mathrm{d}t^{\prime}$. As $A_k$ is the most convoluted term to calculate, let us focus solely on $A_k$, where we find that
\begin{align*}
    A_{k}=&e^{-s}\int_{0}^{s}\sin(\omega t^{\prime})^{k}e^{t^{\prime}}\mathrm{d}t^{\prime}\\
    =&e^{-s}\left[\left.\sin(\omega t^{\prime})^{k}e^{t^{\prime}}\right|_{0}^{s}-k\omega\int_{0}^{s}\cos(\omega t^{\prime})\sin^{k-1}(\omega t^{\prime})e^{t^{\prime}}\mathrm{d}t^{\prime}\right]\\
    =&e^{-s}\left[\sin^{k}(\omega s)e^{s}-k\omega\int_{0}^{s}\cos(\omega t^{\prime})\sin^{k-1}(\omega t^{\prime})e^{t^{\prime}}\mathrm{d}t^{\prime}\right]\\
    =&e^{-s}\left[\sin^{k}(\omega s)e^{s}-k\omega\left(\left.\cos(\omega t^{\prime})\sin^{k-1}(\omega t^{\prime})e^{t^{\prime}}\right|_{0}^{s}-\int_{0}^{s}\left(\omega(k-1)\cos^{2}(\omega t^{\prime})\sin^{k-2}(\omega t^{\prime})-\omega\sin(\omega t^{\prime})^{k}\right)e^{t^{\prime}}\mathrm{d}t^{\prime}\right)\right]\\
    =&e^{-s}\left[\sin^{k}(\omega s)e^{s}-k\omega\cos(\omega s)\sin^{k-1}(\omega s)e^{s}+\omega^2(k-1)k\int_{0}^{s}\cos^{2}(\omega t^{\prime})\sin^{k-2}(\omega t^{\prime})e^{t^{\prime}}\mathrm{d}t^{\prime}-\omega^2k\int_{0}^{s}\sin(\omega t^{\prime})^{k}e^{t^{\prime}}\mathrm{d}t^{\prime}\right]\\
    =&\sin^{k}(\omega s)-k\omega\cos(\omega s)\sin^{k-1}(\omega s)+\omega^2(k-1)ke^{-s}\int_{0}^{s}\cos^{2}(\omega t^{\prime})\sin^{k-2}(\omega t^{\prime})e^{t^{\prime}}\mathrm{d}t^{\prime}-\omega^2kA_k\\
    =&\sin^{k}(\omega s)-k\omega\cos(\omega s)\sin^{k-1}(\omega s)+\omega^2(k-1)ke^{-s}\int_{0}^{s}(1-\sin^{2}(\omega t^{\prime}))\sin^{k-2}(\omega t^{\prime})e^{t^{\prime}}\mathrm{d}t^{\prime}-\omega^2kA_k\\
    =&\sin^{k}(\omega s)-k\omega\cos(\omega s)\sin^{k-1}(\omega s)+\omega^2(k-1)ke^{-s}\int_{0}^{s}\sin^{k-2}(\omega t^{\prime})e^{t^{\prime}}\mathrm{d}t^{\prime}-\omega^2(k-1)ke^{-s}\int_{0}^{s}\sin^{k}(\omega t^{\prime})e^{t^{\prime}}\mathrm{d}t^{\prime}-\omega^2kA_k\\
    =&\sin^{k}(\omega s)-k\omega\cos(\omega s)\sin^{k-1}(\omega s)+\omega^2(k-1)ke^{-s}\int_{0}^{s}\sin^{k-2}(\omega t^{\prime})e^{t^{\prime}}\mathrm{d}t^{\prime}-\omega^2(k-1)kA_{k}-\omega^2kA_{k}.
\end{align*}
We see that we have expressed $A_k$ in terms of itself, which we can easily rearrange to find
\begin{align*}
    A_{k}=&\frac{\sin^{k}(\omega s)-k\omega\cos(\omega s)\sin^{k-1}(\omega s)}{1+\omega^2k+\omega^2(k-1)k}+\frac{\omega^2(k-1)k}{1+\omega^2k+\omega^2(k-1)k}e^{-s}\int_{0}^{s}\sin^{k-2}(\omega t^{\prime})e^{t^{\prime}}\mathrm{d}t^{\prime}\\
    =&\frac{\sin^{k}(\omega s)-k\omega\cos(\omega s)\sin^{k-1}(\omega s)}{1+\omega^2k+\omega^2(k-1)k}+\frac{\omega^2(k-1)k}{1+\omega^2k+\omega^2(k-1)k}A_{k-2}.
\end{align*}
We see that we have now expressed $A_k$ in terms of $A_{k-2}$, allowing us to calculate $A_k$ further using an iterative approach. Before we do so, we can simplify the above expression slightly further by noting that the $\sin^{k}(\omega s)$ always vanishes in the integral term $\int_{nT}^{(n+1/2)T}\cos(\omega s)\sin(\omega s)A_{k}\mathrm{d}s$. To tidy the notation we then introduce the terms
\begin{align*}
    B_{k}(\omega,s)=&\frac{-k\omega\cos(\omega s)\sin^{k-1}(\omega s)}{1+\omega^2k+\omega^2(k-1)k},\\
    C_{k}(\omega)=&\frac{\omega^2(k-1)k}{1+\omega^2k+\omega^2(k-1)k}.
\end{align*}
We can now find our final expression for $A_k$ as follows
\begin{align*}
    A_{k}=&\frac{-k\omega\cos(\omega s)\sin^{k-1}(\omega s)}{1+\omega^2k+\omega^2(k-1)k}+\frac{\omega^2(k-1)k}{1+\omega^2k+\omega^2(k-1)k}A_{k-2}\\
    =&B_{k}(\omega,s)+C_{k}(\omega)A_{k-2}\\
    =&B_{k}(\omega,s)+C_{k}(\omega)(B_{k-2}(\omega,s)+C_{k-2}(\omega)A_{k-4})=B_{k}(\omega,s)+B_{k-2}(\omega,s)C_{k}(\omega)+C_{k}(\omega)C_{k-2}(\omega)A_{k-4}\\
    =&\sum_{i=2}^{(k+1)/2}B_{2i-1}(\omega, s)\prod_{j=i+1}^{(k+1)/2}C_{2j-1}(\omega)+\prod_{n=1}^{(k-1)/2} C_{2n+1}(\omega)e^{-s}\int_{0}^{s}\sin(\omega t^{\prime})\mathrm{d}t^{\prime}\\
    =&\sum_{i=2}^{(k+1)/2}B_{2i-1}(\omega, s)\prod_{j=i+1}^{(k+1)/2}C_{2j-1}(\omega)+\frac{\omega e^{-s}-\omega\cos(\omega s)+\sin(\omega s)}{1+\omega^2}\prod_{n=1}^{(k-1)/2} C_{2n+1}(\omega).
\end{align*}
In the fourth line we used that $k$ is odd and thus the iterative aproach ends with a $\int_{0}^{s}\sin(\omega t^{\prime})\mathrm{d}t^{\prime}$ term (rather than a squared sine). As mentioned before, the $e^{-s}$ term is a transient that cancels out with our choice of $g(0)$ (and would otherwise vanish too), yielding our final expression for $A_k$ as follows
\begin{align*}
    A_{k}=&\sum_{i=2}^{(k+1)/2}B_{2i-1}(\omega, s)\prod_{j=i+1}^{(k+1)/2}C_{2j-1}(\omega)+\frac{-\omega\cos(\omega s)+\sin(\omega s)}{1+\omega^2}\prod_{n=1}^{(k-1)/2} C_{2n+1}(\omega).
\end{align*}
Returning to our expression for $H$, we now find
\begin{align*}
	H=&\left|\omega\sum_{k\in\mathds{O}}\alpha_{k}\left[\int_{nT}^{(n+1/2)T}\cos(\omega s)\sin(\omega s)A_{k}\mathrm{d}s-\int_{(n+1/2)T}^{(n+1)T}\cos(\omega s)\sin(\omega s)A_{k}\mathrm{d}s\right]\right|\\
    =&\left|-\omega\sum_{k\in\mathds{O}}\alpha_{k}\left[\sum_{i=2}^{(k+1)/2}\frac{4\sqrt{\pi}}{(1+2i)(1+(1-2i)^2\omega^2)}\frac{\Gamma\left(i\right)}{\Gamma\left(i-1/2\right)}\prod_{j=i+1}^{(k+1)/2}C_{2j-1}(\omega)+\frac{4}{3\left(1+\omega^2\right)}\prod_{n=1}^{(k-1)/2} C_{2n+1}(\omega)\right]\right|.
\end{align*}
To write the above expression a bit more concisely we introduce
\begin{align*}
    D_i(\omega)=\frac{4\sqrt{\pi}}{(1+2i)(1+(1-2i)^2\omega^2)}.
\end{align*}
This yields our final expression
\begin{align*}
	H=&\left|-\omega\sum_{k\in\mathds{O}}\alpha_{k}\left[\sum_{i=2}^{(k+1)/2}D_i(\omega)\frac{\Gamma\left(i\right)}{\Gamma\left(i-1/2\right)}\prod_{j=i+1}^{(k+1)/2}C_{2j-1}(\omega)+\frac{4}{3\left(1+\omega^2\right)}\prod_{n=1}^{(k-1)/2} C_{2n+1}(\omega)\right]\right|,
\end{align*}
which is Eq.~(\ref{eq:H}) presented in the main text.

%\bibliographystyle{apsrev4-2}
%\bibliography{bibfile}


\end{document}